\documentclass[11pt,oneside]{amsart}

\usepackage{textcomp}
\usepackage[foot]{amsaddr}
\usepackage{geometry}
\usepackage{hyperref}
\usepackage[table]{xcolor}
\usepackage{graphicx}
\usepackage{amsmath}
\usepackage{amssymb}
\usepackage{amsthm}
\usepackage{mathtools}
\usepackage{mathrsfs}
\usepackage{enumerate}
\usepackage{tikz}
\usetikzlibrary{trees,decorations.pathreplacing}
\usepackage[shortalphabetic]{amsrefs}
\usepackage{wrapfig}
\usepackage{gensymb}
\usepackage{braket}
\usepackage[textsize=footnotesize,textwidth=2.7cm]{todonotes}
\usepackage{autonum}
\usepackage{cellspace}

\definecolor{gray}{rgb}{0.95,0.95,0.95}
\DeclareMathOperator{\tr}{tr}

\DeclareMathOperator{\homeo}{Homeo}

\DeclareMathOperator{\Sch}{Sch}

\theoremstyle{plain}
\newtheorem{theorem}{Theorem}[section]
\newtheorem{lemma}[theorem]{Lemma}
\newtheorem{proposition}[theorem]{Proposition}
\newtheorem*{corollary}{Corollary}

\theoremstyle{definition}
\newtheorem{definition}{Definition}[section]

\newtheorem{example}{Example}[section]

\theoremstyle{remark}
\newtheorem*{remark}{Remark}

\DeclarePairedDelimiter\abs{\lvert}{\rvert}

\title[Quantum fields for Thompson's groups]{Quantum fields for unitary representations of Thompson's groups $F$ and $T$}
\author{Tobias J.\ Osborne$^1$}
\address{$^1$Institut für Theoretische Physik, Leibniz Universität Hannover, Appelstraße 2, Hannover 30167, Germany}
\author{Deniz E.\ Stiegemann$^{1,2}$}
\address{$^2$ARC Centre for Engineered Quantum Systems, School of Mathematics and Physics, The University of Queensland, Brisbane, QLD 4072, Australia}

\date{\today}

\begin{document}

\begin{abstract}
We describe how to define observables analogous to quantum fields for the semicontinuous limit recently introduced by Jones in the study of unitary representations of Thompson's groups $F$ and $T$. We find that, in terms of correlation functions of these fields, one can deduce quantities resembling the conformal data, i.e., primary fields, scaling dimensions, and the operator product expansion. Examples coming from quantum spin systems and anyon chains built on the trivalent category $\mathit{SO}(3)_q$ are studied.
\end{abstract}

\maketitle

\section{Introduction}
Complex quantum systems approaching a quantum phase transition present fascinating and nontrivial physics \cite{sachdevQuantumPhaseTransitions2011}. To study such systems a multitude of methods have been developed, the most prominent being to model them via an effective quantum field, enabling the deployment of a multitude of quantum field techniques. When applied to a quantum phase transition one generically expects to obtain a conformal field theory (CFT) \cites{francesco_conformal_1997,cardyConformalInvarianceUniversality1984,cardyOperatorContentTwodimensional1986,bloteConformalInvarianceCentral1986,cardyLogarithmicCorrectionsFinitesize1986,affleckUniversalTermFree1986} as an effective model. This connection has led to a fruitful interplay whereby CFT techniques have led to powerful insights into quantum critical phenomena and, in turn, quantum lattice systems have provided microscopic models for exotic CFTs (see, e.g., \cites{zouConformalFieldsOperator2019,kooRepresentationsVirasoroAlgebra1994a,readAssociativealgebraicApproachLogarithmic2007,dubailConformalFieldTheory2010,gainutdinovLogarithmicConformalField2013,gainutdinovLatticeFusionRules2013,bondesanChiralSUKcurrents2015,milstedExtractionConformalData2017,zouConformalDataRenormalization2018} for a cross section of representative results). Thus it is that physicists regard quantum phase transitions and CFTs as largely synonymous.

At a physical level the approximation of quantum lattice systems via QFTs is a well-established utilitarian procedure and there is a standard lore available to identify the correct QFT modelling a given continuum limit, see e.g., \cites{tsvelikQuantumFieldTheory2007,auerbachInteractingElectronsQuantum1994}. However, quantum lattice systems continue to generate a ready supply of new and ever more fascinating examples challenging standard techniques. A key recent exemplar is the \emph{golden chain}, which is a one-dimensional lattice of Fibonacci anyons \cite{feiguin_interacting_2007}. Progress toward the correct continuum -- or scaling -- limit of this type of system, likely a rational CFT, has most recently been obtained in \cite{ziniConformalFieldTheories2018}, but many conjectures remain unresolved.

Conformal field theory itself, while a powerful tool, is far from a finished research area, both from the physical and mathematical sides. Nevertheless, it does appear that we are slowly converging on a reasonably complete mathematical framework in two spacetime dimensions. One intriguing consequence of recent mathematical investigations into CFTs is the conjecture, supported by the original work of Doplicher \cite{doplicherNewDualityTheory1989} and later by Bischoff \cites{bischoffRelationSubfactorsArising2015,bischoffRemarkCFTRealization2016}, that there is a correspondence between subfactors and CFTs \cite{Jones2010VonNA}. There is now a considerable body of evidence for this conjecture (see, e.g., \cites{xuExamplesSubfactorsConformal2018,calegariCyclotomicIntegersFusion2011} for some recent progress), with the correspondence mapped out apart from a set of most curious examples where there are certain exceptional subfactors with no known counterpart CFT. This most intriguing situation is best exemplified in terms of the Haagerup subfactor \cites{haagerupPrincipal1993,asaeda_exotic_1999}, which is the smallest (finite-depth, irreducible, hyperfinite) subfactor with index more than 4. It is possible to predict \cite{evansExoticnessRealisabilityTwisted2011} properties of the conjectured counterpart Haagerup CFT, but its construction likely goes beyond all known techniques, making it remarkable both as a mathematical construct and as a new physical example.

While there is no known counterpart CFT for the Haagerup subfactor, it is relatively straightforward to write down candidate microscopic models directly built from the corresponding trivalent category H3 \cite{morrison_categories_2015}. From the physical side the challenge now is to identify quantum phase transitions in these models and analyse their properties, i.e., the corresponding central charge etc., to decide if they correspond to something like the conjectured Haagerup CFT, and then to take the scaling or continuum limit around a phase transition. This approach is Jones' ``Royal Road'' \cite{jones_unitary_2014}, and is the most direct attack on the conjecture. Many challenges remain in taking a journey along the royal road, not least of which is the still largely immature status of the theory of quantum phase transitions, and tools therefore.

The search for a direct construction of a counterpart CFT for the Haagerup subfactor was commenced by Jones who initiated a programme \cites{jones_unitary_2014,jones_scale_2017,jones_no-go_2016} to construct continuum limits of lattice systems via a Kadanoff block spin renormalization ansatz. Motivated by the remarkable analogy between Thompson's groups $F$ and $T$ \cite{cannon_introductory_1996} and the conformal group, Jones constructed families of unitary representations for $F$ and $T$ using what physicists would term (nonuniform) \emph{tree tensor networks} (TTN) \cite{bridgeman_hand-waving_2016}. The kinematical Hilbert space described by these tensor networks is called the \emph{semicontinuous limit}. These representations have many striking properties, e.g., amongst others they can lead to knot invariants.

Thompson's groups $F$ and $T$ also play a key role in understanding holographic dualities, particularly the AdS/CFT correspondence, in high energy physics. By taking the semicontinuous limit of the holographic codes of Pastawski, Yoshida, Harlow, and Preskill \cite{pastawski_holographic_2015} one obtains a combinatorial Hilbert space for a boundary theory analogous to a CFT. Dynamics may be then
introduced -- an approach with origins in the work of Penner, Funar, and Sergiescu \cites{penner_universal_1993,funar_central_2010,schneps_geometric_1997} -- by building Jones' unitary representation of $T$ \cite{osborneDynamicsHolographicCodes2017b}. The bulk Hilbert space of the corresponding gravitylike theory is
then realised as a special subspace of the semicontinuous limit spanned distinguished states. The analogue of the group of large bulk diffeomorphisms is then given by a unitary representation of the \emph{Ptolemy group} $\text{Pt}$, on the bulk Hilbert space thus realising a toy model of the AdS/CFT correspondence.

Jones' semicontinuous limit construction has so far been unable to produce new CFTs. This is due to a number of obstructions, perhaps the most serious of which is that the resulting limit is not generically translation invariant \cites{jones_no-go_2016,klieschContinuumLimitsHomogeneous2018}. Nonetheless, the semicontinuous limit does lead to something very much resembling a CFT whose study is interesting in its own right. This is the goal of the present paper: we commence the investigation of what might be termed \emph{Thompson field theory}, explain how to define quantum field-like operators for such representations, calculate their correlation functions, and extract information resembling the conformal data. Very recent related work in this direction may be found in \cites{brothierOperatoralgebraicConstructionGauge2019,osborneContinuumLimitsQuantum2019}.

The material in this paper, which is aimed at both physicists and mathematicians, is presented at correspondingly varying levels of rigour. We consistently warn the reader throughout the main text when material is presented less rigour, for instance, by using words such as ``prototype'' or ``formally''.

Here is a brief guide to the paper:
\begin{itemize}
	\item Section~\ref{sec:preliminaries} contains a reminder of quantum lattice systems and hamiltonians for such lattice systems for both quantum spin systems and anyons in the context of the trivalent category $\mathit{SO}(3)_q$ as well as a brief reminder of Thompson's groups $F$ and $T$.
	\item Section~\ref{sec:corrfunstrees} presents the definition of \emph{ascending operators} and the calculation of the two-point correlation functions for ascending operators for finite regular binary tree states for quantum spin systems.
	\item Section~\ref{sec:renormfieldoprs} may be skipped upon first (and second) reading. This section is \emph{not} presented at a level of mathematical rigour (this is signalled by the presence of ``$\sim$'' symbols towards the end). The objective of this section is to study continuum limits of correlation functions and study what the freedom in rescalings and shifts allows in terms of the existence of a continuum limit. This section may be regarded as the principle motivation for the definitions in section~\ref{sec:primaryfields}.
	\item Section~\ref{sec:modularinvariance} contains a reminder of perfect tensors (and planar perfect tangles) and the unitary representations of Thompson's group $T$ which arise from such boxes. The important property that the ``vacuum state'' for such representations are invariant under $\mathit{PSL}(2,\mathbb{Z})$ is then deduced. This section is not without rigour, and reviews known rigourous material. To save time and space the arguments are presented here via examples.
	\item Section~\ref{sec:primaryfields} is by far the most important section of this paper and contains the definition of (quasi-) primary fields for Thompson's groups $F$ and $T$ as well as a theorem showing how to deduce the correlation functions for these fields. Everything here is rigourous except when we talk about the motivations from quantum field theory. The non-rigorous discussions in this and the following three sections are indicated as such, and no rigourous results depend on them. The results here are only written out for quantum spin systems leaving the generalisation for trivalent categories for the reader. The behaviour of these correlation functions in the continuum are depicted in a couple of representative examples. These plots highlight the core characteristic of the two-point functions for a tree in the continuum, namely, they exhibit scale invariance \emph{and} discontinuities.
	\item Section~\ref{sec:shortdistance} explores some simple yet striking corollaries of the previous section. In particular we explore the short-distance behaviour of $n$-point correlation functions. Except where indicated this section is mathematically rigourous.
	\item Section~\ref{sec:ope} introduces the prototype definition of the operator product expansion. The corresponding (in general, nonassociative) fusion ring is also introduced. Except where indicated this section is mathematically rigourous.
	\item Section~\ref{sec:thompsonaction} strengthens the analogies between Thompson-group symmetric quantum mechanics and CFT. Here the action of the Thompson group on $n$-point functions is deduced culminating in the second main result of the notes, namely (\ref{eq:nptaction}). This section is mathematically rigourous.
	\item Section~\ref{sec:example1} and Section~\ref{sec:example2} contains illustrations of the results of the notes in terms of two important examples: (i) tree states for a quantum spin system; and (ii) a lattice built on cabled $\mathit{SO}(3)_q$.
	\item Section~\ref{sec:em} contains discussion around the challenges facing the definition of a quantity analogous the the energy momentum tensor. 
	\item Appendix~\ref{app:trees} and Appendix~\ref{app:jordan} contain some notes on properties of trees and the Jordan decomposition.
\end{itemize}

\section{Preliminaries}\label{sec:preliminaries}
We work with quantum spin systems or trivalent categories throughout and illustrate results mostly for the quantum spin system built from $\mathbb{C}^d$ or the trivalent category $\mathit{SO}(3)_q$.

\subsection{Quantum spin systems}
Our quantum spin systems $\Lambda_N$ are comprised of a finite number $N$ of quantum spins with local dimension $d$. The quantum spins are assumed to be arranged on a ring. Thus lattice sites are labelled by the integers $j = 0, 1, \ldots, N-1$, with the identification $N= 0$. The total Hilbert space for a quantum spin system of $N$ such spins is therefore
\begin{equation}
	\mathcal{H}_N = (\mathbb{C}^d)^{\otimes N}.
\end{equation}
Observables for this quantum system are Hermitian operators living in
\begin{equation}
	\mathcal{A}(\Lambda_N)= \mathcal{B}(\mathcal{H}_N) = \bigotimes_{j=0}^{N-1} \mathcal{A}_j,
\end{equation}
where $\mathcal{A}_j= M_d(\mathbb{C})$ is the local on-site observable algebra for spin $j$ given by $M_d(\mathbb{C})$, the algebra of $d\times d$ complex matrices.

One can define the observable algebra for an infinite-size $D$-dimensional quantum spin system on a lattice $\mathbb{Z}^D$, $D\in\mathbb{N}$, as the following $C^*$-algebra known as the \emph{quasi-local algebra} \cite{bratteliOperatorAlgebrasQuantum1997}. Firstly, we define, corresponding to any finite subset $\Lambda \subset \mathbb{Z}^D$, the observable algebra $\mathcal{A}(\Lambda)$ to be the tensor product of $\mathcal{A}_j$ over all $j\in \Lambda$. For $\Lambda_1\subset \Lambda_2$ identify $\mathcal{A}(\Lambda_1)$ with the subalgebra $\mathcal{A}(\Lambda_1)\otimes \mathbb{I}_{\Lambda_2\setminus \Lambda_1} \subset \mathcal{A}(\Lambda_2)$. For infinite $\Lambda \subset \mathbb{Z}^D$ denote by $\mathcal{A}(\Lambda)$ the $C^*$-closure of the increasing family of finite-dimensional algebras $\mathcal{A}(\Lambda_f)$ with $\Lambda_f \subset \Lambda$. The quasi-local algebra is then $\mathcal{A}(\mathbb{Z}^D)$.

\subsection{Trivalent Categories}

Trivalent categories are algebraic structures that admit a nice and simple graphical calculus. We will explain the graphical calculus and rules necessary for calculations here and refer to the original paper for the mathematical details \cite{morrison_categories_2015}.

Consider planar trivalent graphs drawn in a rectangle, with $k=n+m$ boundary points, such that $n$ boundary points are located on the top edge of an imaginary rectangle, and $m$ on the bottom edge. We identify graphs related by isotopies that keep the graphs in their rectangles and don't change the order of the boundary points. The graphs thus defined are called diagrams, and the boundary points are referred to as open or uncontracted legs. To indicate the number of open legs on the top and bottom edge of the rectangle, we use the notation $n\to m$.

By allowing formal addition and multiplication by complex scalars, we can turn the set $\mathcal{C}(n\to m)$ of diagrams $n\to m$ into a vector space. We further define two bilinear operations. First, the composition of two diagrams $x\colon n\to m$ and $y\colon m\to l$ is given by stacking $x$ on top of $y$ and connecting the open legs in order, giving the diagram $y\circ x\colon n\to l$. Second, the tensor product of two diagrams $x\colon n\to m$ and $z\colon p\to q$ is given by drawing $x$ and $z$ side by side, giving the diagram $x\otimes z\colon n+p\to m+q$.

\begin{itemize}
	\item the empty diagram $0 \to 0$ containing no vertices and edges,
	\item the diagram $1\to 1$ consisting of a single line:
	\begin{equation}
		\begin{tikzpicture}[scale=0.5]
    	\draw (0, 0) -- (0, 2);
  	\end{tikzpicture}
	\end{equation}
	\item the two graphs $0\to 2$ and $2\to 0$ called cap and cup:
	\begin{equation}
		\begin{tikzpicture}[scale=0.5]
	    \draw (0, 0) arc [start angle=0, end angle=180, radius=1];
	  \end{tikzpicture}\qquad
		\begin{tikzpicture}[scale=0.5]
    	\draw (0, 0) arc [start angle=180, end angle=360, radius=1];
  	\end{tikzpicture}
	\end{equation}
	\item the trivalent vertex $0\to 3$, which is rotation invariant:
	\begin{equation}
		\begin{tikzpicture}[scale=0.5]
    	\draw (0, 2) -- (0, 1);
    	\draw (-1, 0) -- (0, 1) -- (1, 0);
  	\end{tikzpicture}
	\end{equation}
\end{itemize}

Writing $\mathcal{C}_k$ for $\mathcal{C}(0\to k)$, we require that the following dimensional constraints are satisfied:
\begin{equation}
	\dim\mathcal{C}_0=1,\quad \dim\mathcal{C}_1=0,\quad \dim\mathcal{C}_2=1,\quad \dim\mathcal{C}_3=1.
\end{equation}

This has important consequences:
\begin{itemize}
	\item $\dim\mathcal{C}_0=1$ means that every diagram with no open legs is a scalar multiple of the empty diagram. This leads us to identifying $\mathcal{C}_0$ with the underlying field $\mathbb{C}$ by assigning the empty diagram the value $1$. The loop is then a non-zero complex number which we call $d$:
	\begin{equation}\label{eq:loop}
		\begin{tikzpicture}[scale=0.25,baseline=-1mm]
	    \draw (0, 0) circle [radius=1];
	  \end{tikzpicture}
		=d
	\end{equation}
	\item $\dim\mathcal{C}_1=0$ means that every diagram containing the tadpole diagram
	\begin{equation}
		\begin{tikzpicture}[scale=0.25]
	    \draw (0, 0) circle [radius=1];
	    \draw (0, -3) -- (0, -1);
	  \end{tikzpicture}
	\end{equation}
	is zero.
	\item $\dim\mathcal{C}_2=1$ means that
	\begin{equation}
		\begin{tikzpicture}[scale=0.25,baseline=-1mm]
	    \draw (0, 0) circle [radius=1];
	    \draw (0, 1) -- (0, 3);
	    \draw (0, -3) -- (0, -1);
  	\end{tikzpicture}=b\ %
		\begin{tikzpicture}[scale=0.25,baseline=-1mm]
	    \draw (0, -3) -- (0, 3);
	  \end{tikzpicture}
	\end{equation}
	(the diagram on the left-hand side is also called a \emph{bigon}). We choose the normalization $b=1$.
	\item $\dim\mathcal{C}_3=1$ means that
	\begin{equation}\label{eq:triangle}
		\begin{tikzpicture}[scale=0.5,baseline=4mm]
			\draw (0, 2.3) -- (0, 1.5);
	    \draw (-1, 0) -- (-0.5, 0.5);
			\draw (1, 0) -- (0.5, 0.5);
			\draw (0, 1.5) -- (-0.5, 0.5) -- (0.5, 0.5) -- (0, 1.5);
	  \end{tikzpicture}
		= t\, \begin{tikzpicture}[scale=0.5,baseline=4mm]
		\draw (0, 2.3) -- (0, 1);
		\draw (-1, 0) -- (0, 1) -- (1, 0);
	  \end{tikzpicture},
	\end{equation}
	where $t$ is another parameter.
\end{itemize}

Now we make the main assumption: We assume that all elements of $\mathcal{C}(n\to m)$, for any choice of $n$ and $m$, are generated from the basic diagrams using only composition and tensoring.

The relations established in equations (\ref{eq:loop}) to (\ref{eq:triangle}) effectively provide rules to simplify diagrams containing loops, tadpoles, bigons, and triangles, which can be considered as $n$-gons with $n=0, 1, 2, 3$, respectively.

$\mathit{SO}(3)_q$ is the case when $d$ and $t$ satisfy
\begin{equation}
	d+t-d t-2=0,
\end{equation}
and we assume that $d$ is not the golden ratio. In this case we have $\dim\mathcal{C}_4=3$.

\subsection{\texorpdfstring{Thompson's Groups $F$ and $T$}{Thompson's Groups F and T}}
In this subsection we review the definition of Thompson's groups $F$ and $T$. The canonical reference here is \cite{cannon_introductory_1996}. The reader may also find the thesis \cite{belk_thompsons_2007} of Belk to be helpful.

We begin by introducing some definitions. Let $(X,\mathcal{T})$ be a topological space. A \emph{partition} $P$ of $X$ is a collection of disjoint nonempty subsets of $X$ whose union is $X$. Let $P$ and $Q$ be two partitions of $X$, we say that $Q$ is a \emph{refinement} of $P$, denoted $P\preceq Q$, if every element of $Q$ is a subset of an element of $P$. (A useful mnemonic to remember the ordering is that $Q$ has more elements than $P$.) Let $\mathcal{P}$ be a set of partitions of $X$. This is a partially ordered set according to the natural partial order $\preceq$ arising from refinement. Our focus in this paper will be on sets $\mathcal{D}$ of partitions which are \emph{directed} by $\preceq$.

We denote by $\homeo(X)$ the group of homeomorphisms of $X$. Homeomorphisms act in a natural way on partitions: let $f\in\homeo(X)$ and given a partition $P = \{U_1, U_2, \ldots, U_n\}$ of $X$ we obtain another partition of $X$ via
\begin{equation}
	f(P) = \{f(U_j)\,|\, U_j\in P \}.
\end{equation}
This action preserves the partial order $\preceq$.

\begin{definition}
Let $\mathcal{D}$ be a directed set of partitions of $X$. Suppose that $f\in \homeo(X)$. We say that $P\in\mathcal{D}$ is \emph{good} for $f$ if $f(Q) \in \mathcal{D}$, for all $P \preceq Q$, otherwise $P$ is \emph{bad}.
\end{definition}

\begin{example}
	Let $X = [0,1]$ with the usual euclidean topology. Consider the directed set $\mathcal{D}$ of partitions of $[0,1]$ given by intervals with standard dyadic rational endpoints, i.e., each interval (apart from the last) has the form $[a,b)$ where $a = \frac{m}{2^n}$ and $b = \frac{m+1}{2^n}$, with $m,n\in\mathbb{Z}^+$. The final interval is always of the form $[a,1]$ with $a = \frac{m}{2^n}$. Let $f$ be the piecewise linear function given by
	\begin{center}
	\begin{tikzpicture}[scale=3.3]
		\begin{scope}[shift={(1.5,0)}]
			\draw[lightgray] (0.875,0) -- (0.875,1);
			\draw[lightgray] (0.75,0) -- (0.75,1);
			\draw[lightgray] (0.5,0) -- (0.5,1);
			\draw[lightgray] (0, 0.5) -- (1, 0.5);
			\draw[lightgray] (0, 0.625) -- (1, 0.625);
			\draw[lightgray] (0, 0.75) -- (1, 0.75);
			\node[below] at (0.5,-0.07) {$x$};
			\node[below] (a) at (0, -0.05) {$0$};
			\node[below] (b) at (1, -0.05) {$1$};
			\node[left] (d) at (-0.05, 0) {$0$};
			\node[left] (e) at (-0.05, 1) {$1$};
			\draw[lightgray,line width=0.5] (0, 0) rectangle (1,1);
			\draw[->] (0,-0.05) -- (0,1.05);
			\draw[->] (-0.05,0) -- (1.05,0);
			\draw (0, 0) -- (0.5,0.5) -- (0.75,0.625) -- (0.875,0.75) -- (1,1);
		\end{scope}
	\end{tikzpicture}
\end{center}
	and consider the partition $P\in\mathcal{D}$ of the form $P = \{[0,\tfrac{1}{2}), [\tfrac{1}{2},1]\}$. While it may seem that $P$ is good for $f$ because $f(P) = P$, this is not actually the case because the refinement $Q = \{[0,\tfrac{1}{2}), [\tfrac{1}{2},\tfrac{3}{4}), [\tfrac{3}{4},1]\}$, with $P\preceq Q$, is mapped to $f(Q) = \{[0,\tfrac{1}{2}), [\tfrac{1}{2},\tfrac{5}{8}), [\tfrac{5}{8},1]\}$, which is \emph{not} an element of $\mathcal{D}$.
\end{example}

\begin{definition}
	We say that a homeomorphism $f\in\homeo(X)$ is \emph{compatible} with a directed set $\mathcal{D}$ of partitions of $X$ if there exists a partition $P\in \mathcal{D}$ good for $f$.
\end{definition}

\begin{definition}
	We call by \emph{Thompson's group $F$} the group of piecewise linear homeomorphisms from $[0,1]$ to itself which are differentiable except at finitely many dyadic rational numbers and such that on the differentiable intervals the derivatives are powers of $2$.
\end{definition}

\begin{remark}
	That $F$ is indeed a group follows from the following observations. Let $f\in F$. Because the derivative of $f$, where it is defined, is always positive it preserves the orientation of $[0,1]$. Suppose that $0 = x_0 < x_1 < \cdots < x_n = 1$ be the points where $f$ is not differentiable. Then
	\begin{equation}
		f(x) = \begin{cases}
			a_1x, \quad & x_0\le x \le x_1, \\
			a_2x+ b_2, \quad & x_1\le x \le x_2, \\
			&\vdots  \\
			a_nx+ b_n, \quad & x_{n-1}\le x \le x_n,
		\end{cases}
	\end{equation}
	where, for all $j=1, 2, \ldots, n$, $a_j$ is a power of two and $b_j$ is a dyadic rational (where we set $b_1 = 0$). The inverse $f^{-1}$ also has power-of-two derivatives except on dyadic rational points and, since $f$ maps the set of dyadic rationals to itself, we deduce that $F$ is a group under composition.
\end{remark}

\begin{example}
	Consider $X = [0,1]$ with the usual euclidean topology and let $\mathcal{D}$ be the directed set of partitions of $[0,1]$ via intervals with standard dyadic rational endpoints. Then Thompson's group $F$ is compatible with $\mathcal{D}$.
\end{example}

An alternative equivalent characterisation \cite{cannon_introductory_1996} of Thompson's group $F$ is that it is the group of homeomorphisms of $[0,1]$ generated by the piecewise linear functions $A(x)$ and $B(x)$ below, along with their inverses, under composition.
\begin{center}
	\begin{tikzpicture}[scale=3.3]
		\begin{scope}[shift={(0,0)}]
			\draw[lightgray] (0.5,0) -- (0.5,1);
			\draw[lightgray] (0.75,0) -- (0.75,1);
			\draw[lightgray] (0, 0.25) -- (1, 0.25);
			\draw[lightgray] (0, 0.5) -- (1, 0.5);
			\node[below] at (0, -0.05) {$0$};
			\node[below] at (1, -0.05) {$1$};
			\node[left] at (0.0,0.5) {$A(x)$};
			\node[below] at (0.5,-0.07) {$x$};
			\node[left] at (-0.05, 0) {$0$};
			\node[left] at (-0.05, 1) {$1$};
			\draw[lightgray,line width=0.5] (0, 0) rectangle (1,1);
			\draw[->] (0,-0.05) -- (0,1.05);
			\draw[->] (-0.05,0) -- (1.05,0);
			\draw (0, 0) -- (0.5,0.25) -- (0.75,0.5) -- (1,1);
		\end{scope}
		\begin{scope}[shift={(1.5,0)}]
			\draw[lightgray] (0.875,0) -- (0.875,1);
			\draw[lightgray] (0.75,0) -- (0.75,1);
			\draw[lightgray] (0.5,0) -- (0.5,1);
			\draw[lightgray] (0, 0.5) -- (1, 0.5);
			\draw[lightgray] (0, 0.625) -- (1, 0.625);
			\draw[lightgray] (0, 0.75) -- (1, 0.75);
			\node[left] at (0.0,0.5) {$B(x)$};
			\node[below] at (0.5,-0.07) {$x$};
			\node[below] (a) at (0, -0.05) {$0$};
			\node[below] (b) at (1, -0.05) {$1$};
			\node[left] (d) at (-0.05, 0) {$0$};
			\node[left] (e) at (-0.05, 1) {$1$};
			\draw[lightgray,line width=0.5] (0, 0) rectangle (1,1);
			\draw[->] (0,-0.05) -- (0,1.05);
			\draw[->] (-0.05,0) -- (1.05,0);
			\draw (0, 0) -- (0.5,0.5) -- (0.75,0.625) -- (0.875,0.75) -- (1,1);
		\end{scope}
	\end{tikzpicture}
\end{center}

Let $S^1$ be the circle given by $[0,1]$ with $1$ identified with $0$.
\begin{definition}
	We call by \emph{Thompson's group $T$} the group of all piecewise linear homeomorphisms from $S^1$ to itself which take dyadic rational coordinates to dyadic rational coordinates and which are differentiable on $S^1$ except at finitely many dyadic rational coordinates such that when the function is differentiable its derivative is a power of $2$.
\end{definition}

Just as for $F$ one can equivalently realise $T$ as the group generated by $A(x)$ and $B(x)$ above along with $C(x)$ defined below.
\begin{center}
	\begin{tikzpicture}[scale=3.3]
		\draw[lightgray] (0.5,0) -- (0.5,1);
		\draw[lightgray] (0.75,0) -- (0.75,1);
		\draw[lightgray] (0, 0.75) -- (1, 0.75);
		\draw[lightgray] (0, 0.5) -- (1, 0.5);
		\node[below] at (0, -0.05) {$0$};
		\node[below] at (1, -0.05) {$1$};
		\node[left] at (0.0,0.5) {$C(x)$};
		\node[below] at (0.5,-0.07) {$x$};
		\node[left] at (-0.05, 0) {$0$};
		\node[left] at (-0.05, 1) {$1$};
		\draw[lightgray,line width=0.5] (0, 0) rectangle (1,1);
		\draw[->] (0,-0.05) -- (0,1.05);
		\draw[->] (-0.05,0) -- (1.05,0);
		\draw (0, 0.75) -- (0.5, 1);
		\draw (0.5, 0) -- (0.75, 0.5) -- (1, 0.75);
	\end{tikzpicture}
\end{center}

Finally, for Jones' representations of Thompson's groups $F$ and $T$ we will need another characterization, namely as the group of fractions of binary trees. A standard dyadic partition of $[0, 1]$ can be represented by a finite rooted binary tree in which vertices stand for standard dyadic intervals, and the two children of a vertex represent the two subintervals. For example, the partition
\begin{equation}
	\begin{tikzpicture}[scale=5]
		\def\h{0.025}
		\draw (0, 0) -- (1, 0);
		\draw (0, -2*\h) -- ++(0, 4*\h)  node[at start,below=2pt] {$0$};
		\draw (0.25, -\h) -- ++(0, 2*\h) node[at start,below=3pt] {$\frac{1}{4}$};
		\draw (0.375, -\h) -- ++(0, 2*\h) node[at start,below=3pt] {$\frac{3}{8}$};
		\draw (0.5, -\h) -- ++(0, 2*\h) node[at start,below=3pt] {$\frac{1}{2}$};
		\draw (1, -2*\h) -- ++(0, 4*\h) node[at start,below=2pt] {$1$};
	\end{tikzpicture}
\end{equation}
is represented by the tree
\begin{equation}
	\begin{tikzpicture}[scale=0.5]
		\draw (2.0, 0.0) -- (3.0, 1.0) -- (4.0, 0.0);
		\draw (0.0, 0.0) -- (2.0, 2.0) -- (3.0, 1.0);
		\draw (2.0, 2.0) -- (3.0, 3.0) -- (6.0, 0.0);
	\end{tikzpicture}\ .
\end{equation}
Now we consider pairs $(s, t)$ of trees in which both trees have the same number of leaves. We write them in the form
\begin{equation}
	\begin{tikzpicture}[scale=0.25]
		\draw (2.0, 0.0) -- (3.0, 1.0) -- (4.0, 0.0);
		\draw (0.0, 0.0) -- (2.0, 2.0) -- (3.0, 1.0);
		\draw (2.0, 2.0) -- (3.0, 3.0) -- (6.0, 0.0);

		\draw[yshift=-30pt] (-0.5, 0) -- (6.5, 0);

		\begin{scope}[yscale=-1, yshift=60pt]
			\draw (0.0, 0.0) -- (1.0, 1.0) -- (2.0, 0.0);
			\draw (4.0, 0.0) -- (5.0, 1.0) -- (6.0, 0.0);
			\draw (1.0, 1.0) -- (3.0, 3.0) -- (5.0, 1.0);
		\end{scope}
	\end{tikzpicture}
\end{equation}
and consider them as \emph{fractions} in the following way. Two fractions are called equivalent if on can be obtained from the other by removing a pair of opposite \tikz[scale=0.2]{\draw (0, 0) -- (1, 1) -- (2, 0);} carets. For example, the following two fractions are equivalent since the two carets of the left-hand fraction marked bold can be removed to give the fraction on the right-hand side:
\begin{equation}
	\begin{tikzpicture}[scale=0.25,baseline=-3.5mm]
		\draw[very thick] (2.0, 0.0) -- (3.0, 1.0) -- (4.0, 0.0);
		\draw (0.0, 0.0) -- (2.0, 2.0) -- (3.0, 1.0);
		\draw (2.0, 2.0) -- (3.0, 3.0) -- (6.0, 0.0);

		\draw[yshift=-30pt] (-0.5, 0) -- (6.5, 0);

		\begin{scope}[yscale=-1, yshift=60pt]
			\draw[very thick] (2.0, 0.0) -- (3.0, 1.0) -- (4.0, 0.0);
			\draw (3.0, 1.0) -- (4.0, 2.0) -- (6.0, 0.0);
			\draw (0.0, 0.0) -- (3.0, 3.0) -- (4.0, 2.0);
		\end{scope}
	\end{tikzpicture}\,\sim\,
	\begin{tikzpicture}[scale=0.25,baseline=-3.5mm]
		\draw (0.0, 0.0) -- (1.0, 1.0) -- (2.0, 0.0);
		\draw (1.0, 1.0) -- (2.0, 2.0) -- (4.0, 0.0);

		\draw[yshift=-30pt] (-0.5, 0) -- (4.5, 0);

		\begin{scope}[yscale=-1, yshift=60pt]
			\draw (2.0, 0.0) -- (3.0, 1.0) -- (4.0, 0.0);
			\draw (0.0, 0.0) -- (2.0, 2.0) -- (3.0, 1.0);
		\end{scope}
	\end{tikzpicture}
\end{equation}
We denote the equivalence class of a pair $(s, t)$ of trees by $[s, t]$. Next, we define a multiplication operation for (equivalence classes of) fractions. Given two fractions $(s_1, t_1)$ and $(s_2, t_2)$, we can add pairs of carets to obtain fractions $(s_1', t_1')\sim (s_1, t_1)$ and $(s_2', t_2')\sim (s_2, t_2)$ such that $(t_1'=s_2')$. The result of the multiplication is then taken to be $(s_1', t_2')$, and it can be shown that the equivalence class $[s_1', t_2']$ of the result only depends on the \emph{equivalence classes} of the two input fractions, or on the particular choice of the $s_i'$, $t_i'$.

Equivalence classes of fractions form a group under this multiplication (with inverses given by swapping numerators and denominators). This group is isomorphic to Thompson's group $F$. By replacing trees with annular trees (i.e., trees placed on the circle), we get a similar description of Thompson's group $T$.

\section{Correlation functions for tree tensor networks for quantum spin systems}\label{sec:corrfunstrees}
In this section we discuss how to compute correlation functions for tree states defined for \emph{quantum spin systems}. These methods generalise without change to the anyon case modelled by a trivalent category.

Our focus in this section is on the case where
\begin{equation}
	\mathcal{H}_N = (\mathbb{C}^d)^{\otimes N}.
\end{equation}
To define a tree tensor network in such a setting, we need an isometry $V\colon\mathbb{C}^d\to\mathbb{C}^d\otimes\mathbb{C}^d$, which can be represented graphically as a trivalent vertex, or ``$3$-box'',
\begin{equation}
    V \cong \tikz[baseline=8pt, scale=0.75]{
      \draw (0.0, 0.0) -- (0.5, 0.5) -- (1.0, 0.0);
      \draw (0.5, 0.5) -- (0.5, 1.0);
    }\, .
\end{equation}
It considerably simplifies our discussion to assume that
\begin{equation}
	V\textsc{swap} = V,
\end{equation}
where
\begin{equation}
	\textsc{swap}|\phi\rangle|\psi\rangle = |\psi\rangle|\phi\rangle, \quad \forall |\phi\rangle, |\psi\rangle\in \mathbb{C}^d,
\end{equation}
that is,
\begin{equation}
    \tikz[baseline=18pt, scale=0.75]{
      \draw (0.0, 1.0) -- (0.5, 1.5) -- (1.0, 1.0);
      \draw (0.5, 1.5) -- (0.5, 2.0);
      \draw (0.0, 1.0) -- (1.0, 0.0);
      \draw (1.0, 1.0) -- (0.0, 0.0);
    }  = \tikz[baseline=8pt, scale=0.75]{
      \draw (0.0, 0.0) -- (0.5, 0.5) -- (1.0, 0.0);
      \draw (0.5, 0.5) -- (0.5, 1.0);
    }\, ,
\end{equation}
where the cross stands for the $\textsc{swap}$.

Using the $3$-box $V$ we set up the linear map $\mathcal{E}\colon M_d(\mathbb{C}) \rightarrow M_d(\mathbb{C})$ as
\begin{equation}
	\mathcal{E}(X) = V^\dag (X\otimes \mathbb{I}) V,
\end{equation}
which can be represented graphically as
\begin{equation}
	\mathcal{E}(X) =
	\tikz[baseline=28.34pt, scale=0.75]{
      \draw (0.0, 1.0) -- (0.0, 1.2);
      \draw (0.0, 1.5) circle [radius=0.3] node {$X$};
      \draw (0.0, 1.8) -- (0.0, 2.0);
      \draw (0.0, 1.0) -- (0.5, 0.5) -- (1.0, 1.0) -- (1.0, 2.0);
      \draw (0.5, 0.0) -- (0.5, 0.5);
      \draw (0.0, 2.0) -- (0.5, 2.5) -- (1.0, 2.0);
      \draw (0.5, 2.5) -- (0.5, 3.0);
    }\, .
\end{equation}

The map $\mathcal{E}\colon M_d(\mathbb{C}) \rightarrow M_d(\mathbb{C})$ is, by construction, \emph{completely positive}, and admits a Jordan normal form. However, we are going to simplify things and assume further that $\mathcal{E}$ is diagonalisable, so we can obtain right eigenvalues $\lambda_\alpha$ and eigenvectors $\mu^\alpha\in M_d(\mathbb{C})$:
\begin{equation}
	\mathcal{E}(\mu^{\alpha}) =
	\tikz[baseline=28.34pt, scale=0.75]{
      \draw (0.0, 1.0) -- (0.0, 1.1);
      \draw (0.0, 1.5) circle [radius=0.4] node {$\mu^\alpha$};
      \draw (0.0, 1.9) -- (0.0, 2.0);
      \draw (0.0, 1.0) -- (0.5, 0.5) -- (1.0, 1.0) -- (1.0, 2.0);
      \draw (0.5, 0.0) -- (0.5, 0.5);
      \draw (0.0, 2.0) -- (0.5, 2.5) -- (1.0, 2.0);
      \draw (0.5, 2.5) -- (0.5, 3.0);
    } = \lambda_\alpha \tikz[baseline=28.34pt, scale=0.75]{
      \draw (0.0, 0.0) -- (0.0, 1.1);
      \draw (0.0, 1.5) circle [radius=0.4] node {$\mu^\alpha$};
      \draw (0.0, 1.9) -- (0.0, 3.0);
    } = \lambda_\alpha \mu^\alpha, \quad \alpha = 0,1,\ldots, d^2-1.
\end{equation}
Without loss of generality we assume that $\lambda_0 = 1$ and $\mu^0 = \mathbb{I}$. It is worth emphasising that the left eigenvectors $\nu^\alpha$, $\alpha = 0, 1, \ldots, d^2-1$, furnish us with a way to expand -- with respect to the Hilbert-Schmidt inner product $(X,Y) = \frac1d\tr(X^\dag Y)$ -- an operator $M \in M_d(\mathbb{C}^d)$ with respect to $\mu^\alpha$ (see Appendix~\ref{app:jordan} for further details):
\begin{equation}
	M = \sum_{\alpha=0}^{d^2-1} (\nu^\alpha,M) \mu^\alpha.
\end{equation}
The operators $\mu^\alpha$ are called the \emph{scaling} or \emph{ascending operators}.

Now that we have the eigenvalues and eigenvectors of $\mathcal{E}$, we introduce the \emph{fusion map}
\begin{equation}
	\mathcal{F}\colon M_d(\mathbb{C})\times M_d(\mathbb{C}) \rightarrow M_d(\mathbb{C}),
\end{equation}
via
\begin{equation}
	\mathcal{F}(X, Y) = V^\dag(X\otimes Y)V,
\end{equation}
which is diagrammatically given as
\begin{equation}
	\mathcal{F}(X, Y) =\tikz[baseline=28.34pt, scale=0.75]{
      \draw (0.0, 1.0) -- (0.0, 1.2);
      \draw (0.0, 1.5) circle [radius=0.3] node {$X$};
      \draw (0.0, 1.8) -- (0.0, 2.0);
      \draw (1.0, 1.0) -- (1.0, 1.2);
      \draw (1.0, 1.5) circle [radius=0.3] node {$Y$};
      \draw (1.0, 1.8) -- (1.0, 2.0);
      \draw (0.0, 1.0) -- (0.5, 0.5) -- (1.0, 1.0);
      \draw (0.5, 0.0) -- (0.5, 0.5);
      \draw (0.0, 2.0) -- (0.5, 2.5) -- (1.0, 2.0);
      \draw (0.5, 2.5) -- (0.5, 3.0);
    }\, .
\end{equation}
The \emph{fusion coefficients} for this map are given by
\begin{equation}
	{f^{\alpha\beta}}_\gamma = \frac{1}{d}\tr\left((\nu^\gamma)^\dag \mathcal{F}(\mu^\alpha, \mu^\beta)\right)
\end{equation}
so that
\begin{equation}
	\mathcal{F}(\mu^\alpha, \mu^\beta) = \sum_{\gamma} {f^{\alpha\beta}}_\gamma \mu^\gamma.
\end{equation}
The fusion coefficients may be regarded as the structure constants for an, in general, non-associate and non-commutative algebra built on the indices $\alpha$, and
\begin{equation}
	\alpha \star \beta = \sum_{\gamma} {f^{\alpha\beta}}_\gamma \gamma.
\end{equation}

Let $\Omega_N$, $N=2^m$, be the tree state defined via
\begin{center}
	\begin{tikzpicture}[scale=0.35]
    \draw (0.0, 0.0) -- (1.0, 1.0) -- (2.0, 0.0);
    \draw (4.0, 0.0) -- (5.0, 1.0) -- (6.0, 0.0);
    \draw (1.0, 1.0) -- (3.0, 3.0) -- (5.0, 1.0);
    \draw (8.0, 0.0) -- (9.0, 1.0) -- (10.0, 0.0);
    \draw (12.0, 0.0) -- (13.0, 1.0) -- (14.0, 0.0);
    \draw (9.0, 1.0) -- (11.0, 3.0) -- (13.0, 1.0);
    \draw (3.0, 3.0) -- (7.0, 7.0) -- (11.0, 3.0);
    \draw (16.0, 0.0) -- (17.0, 1.0) -- (18.0, 0.0);
    \draw (20.0, 0.0) -- (21.0, 1.0) -- (22.0, 0.0);
    \draw (17.0, 1.0) -- (19.0, 3.0) -- (21.0, 1.0);
    \draw (24.0, 0.0) -- (25.0, 1.0) -- (26.0, 0.0);
    \draw (28.0, 0.0) -- (29.0, 1.0) -- (30.0, 0.0);
    \draw (25.0, 1.0) -- (27.0, 3.0) -- (29.0, 1.0);
    \draw (19.0, 3.0) -- (23.0, 7.0) -- (27.0, 3.0);
    \draw (7.0, 7.0) -- (15.0, 15.0) -- (23.0, 7.0);

    \draw[decoration={brace, amplitude=10}, decorate] (-1.5, -0.5) -- node[left=12pt] {$m$ levels} (-1.5, 15.5);
    \draw[decoration={brace, mirror, amplitude=10}, decorate] (-0.5, -1) -- node[below=12pt] {$N=2^m$ leaves} (30.5, -1);
  \end{tikzpicture}
\end{center}
with $V$ as the trivalent vertex. We first show how to compute one-point correlation functions for the ascending operators, i.e., the expectation values
\begin{equation}
	\langle \mu_j^{\alpha}\rangle_N = \tr(\Omega_N (\mathbb{I}_0\otimes \cdots \otimes \mathbb{I}_{j-1}\otimes \mu_j^\alpha \otimes \mathbb{I}_{j+1}\otimes \cdots \otimes \mathbb{I}_{N-1})).
\end{equation}
This calculation is expedited upon noting that
\begin{equation}\label{eq:expval}
	\langle \mu_j^{\alpha}\rangle_N = \frac1d (\lambda_{\alpha})^{m-1} \tr(\mu^\alpha) = (\lambda_{\alpha})^{m-1} (\mathbb{I},\mu^\alpha).
\end{equation}
For example, for $m=3$, $N=2^3=8$, and $j=2$, we have
\begin{center}
	\begin{tikzpicture}[scale=0.5]
		\draw (0.0, 0.0) -- (1.0, 1.0) -- (2.0, 0.0);
		\draw (4.0, 0.0) -- (5.0, 1.0) -- (6.0, 0.0);
		\draw (1.0, 1.0) -- (3.0, 3.0) -- (5.0, 1.0);
		\draw (8.0, 0.0) -- (9.0, 1.0) -- (10.0, 0.0);
		\draw (12.0, 0.0) -- (13.0, 1.0) -- (14.0, 0.0);
		\draw (9.0, 1.0) -- (11.0, 3.0) -- (13.0, 1.0);
		\draw (3.0, 3.0) -- (7.0, 7.0) -- (11.0, 3.0);

		\foreach \i in {0,1,3,4,5,6,7}{
			\draw (2*\i, 0) -- ++(0, -3.4);
		}

		\draw (4, 0) -- (4, -1);
		\draw (4, -1.7) circle[radius=0.7] node {$\mu^\alpha$};
		\draw (4, -2.4) -- (4, -3.4);

		\begin{scope}[yscale=-1,yshift=3.4cm]
			\draw (0.0, 0.0) -- (1.0, 1.0) -- (2.0, 0.0);
			\draw (4.0, 0.0) -- (5.0, 1.0) -- (6.0, 0.0);
			\draw (1.0, 1.0) -- (3.0, 3.0) -- (5.0, 1.0);
			\draw (8.0, 0.0) -- (9.0, 1.0) -- (10.0, 0.0);
			\draw (12.0, 0.0) -- (13.0, 1.0) -- (14.0, 0.0);
			\draw (9.0, 1.0) -- (11.0, 3.0) -- (13.0, 1.0);
			\draw (3.0, 3.0) -- (7.0, 7.0) -- (11.0, 3.0);
		\end{scope}
	\end{tikzpicture}
\end{center}

Building on this we next focus on the goal of computing the two-point correlation function
\begin{equation}
	\langle \mu^\alpha_j \mu^\beta_k\rangle_N = \tr(\Omega_N [\mathbb{I}_0\otimes \cdots \otimes \mathbb{I}_{j-1}\otimes \mu_j^\alpha \otimes \mathbb{I}_{j+1}\otimes \cdots \otimes \mathbb{I}_{k-1}\otimes \mu^{\beta}_k \otimes \mathbb{I}_{k+1}\otimes \cdots \otimes \mathbb{I}_{N-1}]),
\end{equation}
for $0\le j<k < 2^m-1$. To this end we label the leaves of the regular binary tree $\mathcal{T}_m$ having $2^m$ leaves with binary expansions as per Appendix~\ref{app:trees}. In this way we associate to the $j$-th vertex, $j\in \{0,1,\ldots, 2^m-1\}$, the number $x= j/{2^m}\in [0,1]$. We write
\begin{equation}
	C^{\alpha\beta}_m(x,y) = \langle \mu^\alpha_j \mu^\beta_k\rangle_N,
\end{equation}
where $x = j/{2^m}$ and $y = k/{2^m}$.

The key to computing two-point correlation functions is to note that we can relate $C^{\alpha\beta}(x,y)$ to $C^{\alpha\beta}(x^{(1)},y^{(1)})$, where we've employed the notation $x^{(1)}$ in Appendix~\ref{app:trees} for the number whose binary expansion has one fewer digit than that for $x$. From Appendix~\ref{app:trees} we also take the tree metric $d_T$.

\begin{lemma}
	Let $x, y\in [0,1)$ be two points with $m$-digit binary expansions
	\begin{equation}
		x = 0.x_{-1}\cdots x_{-m}, \qquad y = 0.y_{-1}\cdots y_{-m}.
	\end{equation}
	 Then, writing $d_T(x,y) = k+1$, we have that
	\begin{equation}
		C^{\alpha\beta}_m(x,y) = (\lambda_{\alpha}\lambda_{\beta})^k C^{\alpha\beta}_{m-k}(x^{(k)},y^{(k)}).
	\end{equation}
\end{lemma}
\begin{proof}
	The details are omitted because the basic argument is a relatively straightforward induction. Since $\mu^\alpha_j$ and $\mu^\beta_k$ are on different carets of the tree we can independently apply $\mathcal{E}$ to these operators to relate the expectation values at different levels:
	\begin{equation}
		C^{\alpha\beta}_m(x,y) = \langle \mu^\alpha_j \mu^\beta_k\rangle_N = \langle \mathcal{E}(\mu^\alpha_j) \mathcal{E}(\mu^\beta_k)\rangle_{N/2} = (\lambda_{\alpha}\lambda_{\beta})C^{\alpha\beta}_{m-1}(x^{(1)},y^{(1)}).
	\end{equation}
	The two operators only meet at a caret at level $m-k$, so we can repeat the above process $k$ times.
\end{proof}

Now we study what happens when $d_T(x,y) = 1$. In this case we can relate the two-point correlation function to an expectation value.
\begin{lemma}
	Suppose we have two leaves labelled $x$ and $y$ and that $d_T(x,y) = 1$. Then
	\begin{equation}
		C^{\alpha\beta}_m(x,y) = \tr(\Omega_{N/2}\mathcal{F}(\mu^\alpha, \mu^\beta)).
	\end{equation}
	Exploiting the formula (\ref{eq:expval}) for the expectation value we hence obtain
	\begin{equation}
		C^{\alpha\beta}_m(x,y)  = \frac1d  \sum_{\gamma = 0}^{d^2-1} {f^{\alpha\beta}}_\gamma (\lambda_{\gamma})^{m-2} \tr(\mu^\gamma).
	\end{equation}
\end{lemma}

Putting these lemmas all together, along with Lemma~\ref{lem:treemetric}, we obtain the following formula for the two-point correlation function
\begin{proposition}
	Let $x, y\in [0,1)$ be two points with $m$-digit binary expansions
	\begin{equation}
		x = 0.x_{-1}\cdots x_{-m}, \qquad y = 0.y_{-1}\cdots y_{-m}.
	\end{equation}
	Then
	\begin{equation}
		C^{\alpha\beta}_m(x,y) = \frac1d (\lambda_\alpha\lambda_\beta)^{m+\lfloor \log_2(y\ominus x) \rfloor} \sum_{\gamma = 0}^{d^2-1} {f^{\alpha\beta}}_\gamma (\lambda_{\gamma})^{-\lfloor \log_2(y\ominus x) \rfloor-2} \tr(\mu^\gamma).
	\end{equation}
\end{proposition}
A particularly important special case is the ``end-to-end'' correlation function between the leaf labelled $0.00\cdots0$ and the leaf labelled $0.11\cdots 1$. This may be computed as
\begin{corollary}
	Let $x, y\in [0,1)$ be two points with $m$-digit binary expansions
	\begin{equation}
		x = 0.0_{-1} \cdots 0_{-m}, \qquad y = 0.1_{-1}\cdots 1_{-m}.
	\end{equation}
	Then
	\begin{equation}
		C^{\alpha\beta}_m(x,y) = \frac1d (\lambda_\alpha\lambda_\beta)^{m-1} \sum_{\gamma = 0}^{d^2-1} {f^{\alpha\beta}}_\gamma (\lambda_{\gamma})^{-1} \tr(\mu^\gamma).
	\end{equation}
\end{corollary}

\section{Renormalised field operators}\label{sec:renormfieldoprs}

In this section, following the arguments of \cite{osborneContinuumLimitsQuantum2019}, the objective is to interpret the two-point correlation functions as being for some kind of \emph{discretised quantum field operators}. To do this requires the introduction of a ``field-strength renormalization''. We provide the main motivation for our later definition of the \emph{primary field operators}. As this section has mostly motivational character, it can be skipped upon first reading.

We want to think of the eigenvectors $\mu^\alpha$ of the channel $\mathcal{E}$ as the discretisation of some putative primary field operators on a lattice. To this end we define the \emph{discretised field operators on a lattice with spacing $a_m = a^{m}$}, where $a > 0$. Thus we introduce a new correspondence between leaves $j\in\{0,1,\ldots, 2^m-1\}$ of the regular binary tree $\mathcal{T}_m$ and points $x_j\in\mathbb{R}$ on the real line:
\begin{equation}
	j \leftrightarrow x_j = f_m(j) = a_m (j - 2^{m-1}).
\end{equation}
In this way, given $x_j$, we obtain $j$ according to
\begin{equation}
	j = a_m^{-1} x_j + 2^{m-1}.
\end{equation}
The choice of lattice spacing that corresponds to the notation of the previous section is $a=\tfrac12$. We hence define the discretised field operator at scale $m$:
\begin{equation}
	\phi_\alpha (x_j) = Z_\alpha(m)\mu_{a_m^{-1} x_j + 2^{m-1}}^{\alpha},
\end{equation}
where $x_j \in [-a_m2^{m-1}, a_m2^{m-1}]$. The number $Z_\alpha(m)$ is called the \emph{field-strength renormalization} and it is allowed to depend on the \emph{scale} or \emph{level} $m$.

The process whereby we determine $Z_\alpha(m)$ is as follows. We first compute the \emph{correlation function at scale $m$} of the discretised primary fields via
\begin{equation}
	\langle \phi_\alpha (x)\phi_\beta (y)\rangle_m = Z_\alpha(m)Z_\beta(m)\langle \mu^\alpha_{a_m^{-1} x + 2^{m-1}} \mu^\beta_{a_m^{-1} y + 2^{m-1}}\rangle_N,
\end{equation}
i.e., they are simply the rescaled two-point correlation functions. The next step is then to try and take the \emph{continuum limit} $m\rightarrow \infty$ of these correlation functions to define the symbols
\begin{equation}
	\langle \hat{\phi}_\alpha (x)\hat{\phi}_\beta (y)\rangle \text{``$=$''} \lim_{m\rightarrow \infty}\langle \phi_\alpha (x)\phi_\beta (y)\rangle_m.
\end{equation}
Here the scare quotes indicate that the desired limits might not exist. The art is then to choose the field-strength renormalizations so that these limits (or some subset thereof) exist. This forces consistency conditions onto the $Z_\alpha(m)$s which are described below.

We typically choose
\begin{equation}
	Z_\alpha(m) = Z_\alpha^m,
\end{equation}
for some constants $Z_\alpha$, however, this is not strictly necessary.

Before we collect together the results in a proposition, let's explore the correlator
\begin{equation}
	\langle \phi_\alpha (x)\phi_\beta (y)\rangle_m.
\end{equation}
Thanks to the results of the previous section we find
\begin{multline}
	\langle \phi_\alpha (x)\phi_\beta (y)\rangle_m = \frac1d (Z_\alpha Z_\beta)^m(\lambda_\alpha\lambda_\beta)^{d_T(a_m^{-1} x + 2^{m-1},a_m^{-1} y + 2^{m-1})-1} \times \\ \sum_{\gamma = 0}^{d^2-1} {f^{\alpha\beta}}_\gamma (\lambda_{\gamma})^{m-d_T(a_m^{-1} x + 2^{m-1},a_m^{-1} y + 2^{m-1})-1} \tr(\phi^\gamma).
\end{multline}
Using the fact that
\begin{equation}
	d_T(a_m^{-1} x + 2^{m-1},a_m^{-1} y + 2^{m-1}) = m+1+\left\lfloor \log_2 \left( [(2a)^{-m}x + \tfrac12]\ominus[(2a)^{-m}y + \tfrac12] \right) \right\rfloor
\end{equation}
so that
\begin{multline}
	\langle \phi_\alpha (x)\phi_\beta (y)\rangle_m = \frac1d (Z_\alpha Z_\beta)^m(\lambda_\alpha\lambda_\beta)^{m+\left\lfloor \log_2 \left( [(2a)^{-m}x + \frac12]\ominus[(2a)^{-m}y + \frac12] \right) \right\rfloor} \times \\ \sum_{\gamma = 0}^{d^2-1} {f^{\alpha\beta}}_\gamma (\lambda_{\gamma})^{-\left\lfloor \log_2 \left( [(2a)^{-m}x + \frac12]\ominus[(2a)^{-m}y + \frac12] \right) \right\rfloor-2} \tr(\phi^\gamma).
\end{multline}
At this point we invoke physical arguments: we employ the ``estimate'' $\lfloor\log_2(y\ominus x)\rfloor \sim \log_2(y-x)$ by writing $\lfloor\log_2(y\ominus x)\rfloor=\log_2(y-x) + \Delta(y-x)$.
\begin{multline}
	\langle \phi_\alpha (x)\phi_\beta (y)\rangle_m = \frac1d (Z_\alpha Z_\beta)^m(\lambda_\alpha\lambda_\beta)^{m+ \log_2 \left( (2a)^{-m}(y-x)\right) + \Delta\left((2a)^{-m}(y-x))\right) } \times \\ \sum_{\gamma = 0}^{d^2-1} {f^{\alpha\beta}}_\gamma (\lambda_{\gamma})^{-\log_2 \left( (2a)^{-m}(y-x)\right) -\Delta\left((2a)^{-m}(y-x))\right)-2} \tr(\phi^\gamma).
\end{multline}
Although it is not justified to do so we now neglect $\Delta(y-x)$ to obtain an initial estimate for the behaviour of the correlator, in the hope that we can work out a condition for the continuum limits to exist:
\begin{align}
	&\langle \phi_\alpha (x)\phi_\beta (y)\rangle_m\\
	&\sim \frac1d (Z_\alpha Z_\beta)^m(\lambda_\alpha\lambda_\beta)^{m+ \log_2 \left( (2a)^{-m}(y-x)\right) } \sum_{\gamma = 0}^{d^2-1} {f^{\alpha\beta}}_\gamma (\lambda_{\gamma})^{-\log_2 \left( (2a)^{-m}(y-x)\right)-2} \tr(\phi^\gamma) \\
	&= \frac1d \sum_{\gamma = 0}^{d^2-1} {f^{\alpha\beta}}_\gamma (Z_\alpha Z_\beta)^m(\lambda_\alpha\lambda_\beta)^{m+ \log_2 \left( (2a)^{-m}(y-x)\right) }(\lambda_{\gamma})^{-\log_2 \left( (2a)^{-m}(y-x)\right)-2} \tr(\phi^\gamma) \\
	&= \frac1d \sum_{\gamma = 0}^{d^2-1} {f^{\alpha\beta}}_\gamma (Z_\alpha Z_\beta)^m(\lambda_\alpha\lambda_\beta)^{m-m\log_2 (2a) + \log_2(y-x)}(\lambda_{\gamma})^{m\log_2 (2a) - \log_2(y-x)-2} \tr(\phi^\gamma) \\
	&= \frac1d\sum_{\gamma = 0}^{d^2-1} {f^{\alpha\beta}}_\gamma \left(\frac{Z_\alpha Z_\beta}{\lambda_\gamma}\right)^m\left(\frac{\lambda_\alpha\lambda_\beta}{\lambda_\gamma}\right)^{-m\log_2 (a)}(y-x)^{h_\alpha+h_\beta - h_\gamma} \lambda_\gamma^{-2}\tr(\phi^\gamma) \\
	&= \frac1d\sum_{\gamma = 0}^{d^2-1} {f^{\alpha\beta}}_\gamma \left(\frac{Z_\alpha a^{h_\alpha} Z_\beta a^{h_\beta}}{\lambda_\gamma a^{h_\gamma}}\right)^m(y-x)^{h_\alpha + h_\beta - h_\gamma} \lambda_\gamma^{-2}\tr(\phi^\gamma)
\end{align}
where
\begin{equation}
	h_\alpha = -\log_2(\lambda_\alpha).
\end{equation}
Note that, being generically complex, $h_\alpha$ may have some horrible oscillating phase.

By isolating the $m$-dependent behaviour we obtain the continuum limit criteria:
\begin{equation}\label{eq:ctslimitcriteria}
	|Z_\alpha  Z_\beta | \le | \lambda_\gamma ||a^{h_\gamma-h_\alpha-h_\beta}|, \quad \forall \alpha,\beta,\gamma.
\end{equation}
Define the \emph{renormalised structure constants}
\begin{equation}
	{\widetilde{f}^{\alpha\beta}\hspace{0pt}}_\gamma = \lim_{m\rightarrow\infty}{f^{\alpha\beta}}_\gamma \left(\frac{Z_\alpha a^{d_\alpha} Z_\beta a^{d_\beta}}{\lambda_\gamma a^{d_\gamma}}\right)^m.
\end{equation}
Depending on the values we choose for $a$ and $Z_\alpha$ we have three (well, four) possibilities for the renormalised structure constants. Either: (a) ${\widetilde{f}^{\alpha\beta}\hspace{0pt}}_\gamma =0$; (b) ${\widetilde{f}^{\alpha\beta}\hspace{0pt}}_\gamma = {f^{\alpha\beta}}_\gamma$; or (c) ${f^{\alpha\beta}}_\gamma = \infty$. The third option is regarded as \emph{unphysical} and is excluded by condition (\ref{eq:ctslimitcriteria}). (The fourth option is that the limit doesn't exist because term in the brackets has absolute value $1$ but oscillates as $m$ increases.) The remaining two options then furnish an \emph{In\"{o}n\"{u}-Wigner contraction} of the structure constants ${f^{\alpha\beta}}_\gamma$.

The option that we take is to set $a=1/2$ and
\begin{equation}\label{eq:ffield}
  Z_\alpha = 2^{h_\alpha}.
\end{equation}

\section{\texorpdfstring{Modular invariance of Jones' unitary representations of $T$ arising from perfect tensors}{Modular invariance of Jones' unitary representations of T arising from perfect tensors}}\label{sec:modularinvariance}
In this section we discuss the unitary representations of $T$ arising from special 3-boxes $V$ called planar perfect tangles. We will see that such tangles give rise to representations of $T$ which are invariant under the action of the modular group $PSL(2,\mathbb{Z})$, which is a subgroup of $T$. This property is a coarser analogue of \emph{global conformal invariance} for CFTs.

\subsection{Jones' unitary representations}

First we review a special case of Jones' unitary representations of Thompson's groups \cites{jones_unitary_2014,jones_no-go_2016}.
Consider the set $\mathcal{T}$ of all binary trees. There is an order relation $\preceq$ on $\mathcal{T}$ given by inclusion of rooted trees, turning $\mathcal{T}$ into a directed set. A directed system, in this context, is a collection of Hilbert spaces $\mathcal{H}_t$, one for each tree $t$, and maps $T^s_t\colon\mathcal{H}_s\to\mathcal{H}_t$ such that
\begin{equation}
	T^s_u=T^t_u T^s_t
\end{equation}
whenever $s\preceq t\preceq u$, and $T^s_s=\mathbb{I}$ for all $s\in\mathcal{T}$.

The directed system has a direct limit, constructed as follows. Let
\begin{equation}
	\hat{\mathcal{H}}=\coprod_{t\in\mathcal{T}}\mathcal{H}_t
\end{equation}
be the disjoint union. On the vector space $\hat{\mathcal{H}}$ we define an equivalence relation $\sim$ as follows. Let $s, t\in\mathcal{T}$, and let $\phi_s\in\mathcal{H}_s$ and $\psi_t\in\mathcal{H}_t$. We have
\begin{equation}
	\phi_s\sim\psi_t
\end{equation}
if and only if there exists $u\in\mathcal{T}$ with $s, t\preceq u$ such that
\begin{equation}
	T^s_u \phi_s = T^t_u \psi_t.
\end{equation}
The quotient $\hat{\mathcal{H}}/{\sim}$ has an inner product, which for two vectors $\phi_s\in\mathcal{H}_s$ and $\psi_t\in\mathcal{H}_t$ is given by
\begin{equation}
	\langle T^s_u\phi_s, T^t_u\psi_t\rangle,
\end{equation}
where $u$ with $s, t\preceq u$ exists since $\mathcal{T}$ is a directed set. One can show that the inner product is well-defined. The completion $\mathcal{H}$ of $\hat{\mathcal{H}}/{\sim}$ with respect to the norm induced by the inner product is a well-defined Hilbert space, referred to as the direct limit of the direct system.

Next, we come to Jones' construction of representations of Thompson's groups $F$ (and $T$) on the direct limit Hilbert space. Recall that elements of Thompson's group $F$ can be considered as fractions of binary trees. A fraction $(s, t)$ can be extended by (possibly multiple) pairs of opposing carets to give a new, equivalent fraction $(s', t')$. We can view attachment of carets as an associative multiplication acting on trees, so that in the previous example the added carets are contained in a so-called forest $p$ such that we can write $(s', t')=(ps,pt)$. Now assume that the whole direct limit is built out of a single chosen Hilbert space $\mathcal{H}=\mathbb{C}^d$ together with an isometry $V\colon\mathcal{H}\to\mathcal{H}\otimes\mathcal{H}$. This means that for every tree $t$,
\begin{equation}
	\mathcal{H}_t=\bigotimes_{\abs{t}} \mathcal{H},
\end{equation}
where $\abs{t}$ denotes the number of leaves of $t$. Furthermore, if $s\preceq t$ and $ps=t$, that is, $p$ extends $s$ to give $t$, then we define $T^s_t$ to be the linear map $\Phi_V(p)\colon\mathcal{H}_s\to\mathcal{H}_t$ obtained by replacing each caret in $p$ by an instance of $V$. In this setting, we define a representation $\pi$ of $F$ on the direct limit $\mathcal{H}$ as follows. If $(s, t)$ is a fraction and $\phi_u\in\mathcal{H}_u$ an arbitrary vector in $\hat{\mathcal{H}}$, let $p$ and $q$ be such that $pt=qu$. Then we define
\begin{equation}
	\pi\bigl( [s, t] \bigr) [\phi_u] = [(\Phi(q)\phi)_{ps}],
\end{equation}
where the square brackets again denote taking equivalence classes. Jones showed that this gives a well-defined unitary representation $\pi$. The whole construction works similar for $T$ by replacing trees with annular trees.

There is always a particularly important state, called the \emph{vacuum state}, whose name is justified in Section~\ref{sub:modular}. It is the equivalence class of the state
\begin{equation}
	|\Omega_2\rangle =\frac{1}{\sqrt d} \sum_{j=0}^{d-1} |jj\rangle\in\mathcal{H}_{\tikz[scale=0.08]{\draw (0, 0) -- (1, 1) -- (2, 0);}},
\end{equation}
where $|j\rangle$ denotes any orthonormal basis of $\mathcal{H}$. The representatives of $|\Omega\rangle=[|\Omega_2\rangle]_{\sim}$ correspond to binary trees whose vertices are copies of $V$.

\subsection{Three-leg perfect tensors}

The notion of a perfect tensor was introduced by Pastawski, Yoshida, Harlow, and Preskill \cite{pastawski_holographic_2015}. These highly nongeneric objects capture a discrete version of \emph{rotation invariance} which is extremely useful in building network approximations to continuous manifolds. Since in this work we only deal with binary trees and thus all vertices are trivalent, we only define trivalent perfect tensors, and refer to \cite{pastawski_holographic_2015} for the general definition.

\begin{definition}
	A tensor $V\colon\mathbb{C}^d\to\mathbb{C}^d\otimes\mathbb{C}^d$ is called \emph{planar perfect} if it is proportional to an isometry from $\mathbb{C}^d$ to $\mathbb{C}^d\otimes\mathbb{C}^d$ for all possible choices of in and out legs; in other words, we have
	\begin{equation}
		\def\x{1.5}
		\begin{tikzpicture}[scale=0.6,baseline=-3mm]
			\draw (0, 0.5*\x) -- (0, 1.5*\x);
			\draw (0, 0.5*\x) -- (-1.5/1.73205*\x, 0) -- ++(0, -0.5*\x);
			\draw (0, 0.5*\x) -- (1.5/1.73205*\x, 0) -- ++(0, -0.5*\x);
			\node[above right] at (0, 0.5*\x) {$V^\dag$};
			\begin{scope}[yscale=-1, yshift=0.5*\x cm]
				\draw (0, 0.5*\x) -- (0, 1.5*\x);
				\draw (0, 0.5*\x) -- (-1.5/1.73205*\x, 0);
				\draw (0, 0.5*\x) -- (1.5/1.73205*\x, 0);
				\node[below left] at (0, 0.5*\x) {$V$};
			\end{scope}
		\end{tikzpicture}\,\sim\,
		\begin{tikzpicture}[scale=0.6,baseline=-3mm]
			\draw (0, -2*\x) -- (0, 1.5*\x);
		\end{tikzpicture}\mspace{90mu}
		\begin{tikzpicture}[scale=0.6,baseline=-3mm]
			\draw (0, 0.5*\x) -- (0, 1.5*\x);
			\draw (0, 0.5*\x) -- (-1.5/1.73205*\x, 0) -- ++(0, -0.5*\x);
			\draw (0, 0.5*\x) -- (1.5/1.73205*\x, 0) -- ++(0, -0.5*\x);
			\node[above left] at (0, 0.5*\x) {$V^\dag$};
			\begin{scope}[yscale=-1, yshift=0.5*\x cm]
				\draw (0, 0.5*\x) -- (0, 1.5*\x);
				\draw (0, 0.5*\x) -- (-1.5/1.73205*\x, 0);
				\draw (0, 0.5*\x) -- (1.5/1.73205*\x, 0);
				\node[below right] at (0, 0.5*\x) {$V$};
			\end{scope}
		\end{tikzpicture}\,\sim\,
		\begin{tikzpicture}[scale=0.6,baseline=-3mm]
			\draw (0, -2*\x) -- (0, 1.5*\x);
		\end{tikzpicture}\mspace{90mu}
		\begin{tikzpicture}[scale=0.6,baseline=-3mm]
			\draw (0, 0.5*\x) -- (0, 1.5*\x);
			\draw (0, 0.5*\x) -- (-1.5/1.73205*\x, 0) -- ++(0, -0.5*\x);
			\draw (0, 0.5*\x) -- (1.5/1.73205*\x, 0) -- ++(0, -0.5*\x);
			\node[below] at (0, 0.5*\x) {$V^\dag$};
			\begin{scope}[yscale=-1, yshift=0.5*\x cm]
				\draw (0, 0.5*\x) -- (0, 1.5*\x);
				\draw (0, 0.5*\x) -- (-1.5/1.73205*\x, 0);
				\draw (0, 0.5*\x) -- (1.5/1.73205*\x, 0);
				\node[above] at (0, 0.5*\x) {$V$};
			\end{scope}
		\end{tikzpicture}\,\sim\,
		\begin{tikzpicture}[scale=0.6,baseline=-3mm]
			\draw (0, -2*\x) -- (0, 1.5*\x);
		\end{tikzpicture}
	\end{equation}
	where the straight line represents the identity and the first leg of $V$ follows the label in counter-clockwise direction.
\end{definition}

In addition to this, we will always require that planar perfect tensors are rotation invariant. An example for qutrits is given by $V\colon\mathbb{C}^3\to\mathbb{C}^3\otimes\mathbb{C}^3$,
\begin{equation}\label{eq:perfect-example}
	\langle jk|V|l\rangle = \begin{cases}
		0 & \text{if $j=k$, $k=l$, or $l=j$,}\\
		\frac{1}{\sqrt{2}} & \text{otherwise.}
	\end{cases}
\end{equation}

\subsection{Modular invariance of the vacuum}\label{sub:modular}
Here we show that $|\Omega\rangle$ is invariant under the modular group $\mathit{PSL}(2,\mathbb{Z})$ with presentation
\begin{equation}
	\langle a,b|a^2=(ab)^3=1\rangle.
\end{equation}
To see this note that $\mathit{PSL}(2,\mathbb{Z})$ is a subgroup of $T$ under the homomorphism $\phi(a) = S$ and $\phi(b) = S^{-1}C$, where
\begin{equation}
	S(x) = \begin{cases}
		x+\frac12, \quad x\in [0,\tfrac12) \\ x-\frac12, \quad x\in [\tfrac12,1)
	\end{cases}
\end{equation}
is depicted below:
\begin{center}
	\begin{tikzpicture}[scale=3.3]
		\begin{scope}[shift={(3,0)}]
			\draw[lightgray] (0.5,0) -- (0.5,1);
			\draw[lightgray] (0, 0.5) -- (1, 0.5);
			\node[below] (a) at (0, -0.05) {$0$};
			\node[below] (b) at (1, -0.05) {$1$};
			\node[left] (d) at (-0.05, 0) {$0$};
			\node[left] (e) at (-0.05, 1) {$1$};
			\node[left] at (0.0,0.5) {$S(x)$};
			\node[below] at (0.5,-0.07) {$x$};
			\draw[lightgray,line width=0.5] (0, 0) rectangle (1,1);
			\draw[->] (0,-0.05) -- (0,1.05);
			\draw[->] (-0.05,0) -- (1.05,0);
			\draw (0, 0.5) -- (0.5,1);
			\draw (0.5, 0) -- (1,0.5);
		\end{scope}
	\end{tikzpicture}
\end{center}
We have $S\in T$ since $S=AC$. One may verify that $\pi(S)|\Omega\rangle = |\Omega\rangle = \pi(C)|\Omega\rangle$ (in general, $V$ needs to be rotation invariant for the second equality). We hence deduce that the vacuum state furnishes a one-dimensional representation of $\mathit{PSL}(2,\mathbb{Z})$, i.e., it is invariant.

A striking consequence of $\mathit{PSL}(2,\mathbb{Z})$ invariance is $\mathbb{Z}$ invariance, which we can intepret as translation invariance. The subgroup $\mathbb{Z}$ acts via the M\"obius transformation
\begin{equation}\label{eq:moebius}
  z \mapsto \frac{az+b}{cz+d} = z+n, \quad n\in\mathbb{Z},
\end{equation}
with $a = 1$, $b=n$, $c = 0$, and $d= 1$. Now if we regard the circle $S^1$ as the real line with infinity $\mathbb{R}\cup\{\infty\}$ (under a Cayley transformation and the Minkowski question mark $?(x)$) then we see that the modular subgroup $\mathit{PSL}(2,\mathbb{Z})$ of Thompson's group $T$ acts on $\mathbb{R}\cup\{\infty\}$ precisely via (\ref{eq:moebius}), i.e., integer shifts. Such a transformation is \emph{parabolic}, i.e., it has one fixed point on the circle, namely $\infty$.

Based on these observations we are happy to conjecture that the vacuum state is invariant under $\mathit{PSL}(2,\mathbb{Z})$ if and only if $V$ is a planar perfect tangle.

\section{\texorpdfstring{The continuum limit: primary fields for Thompson's groups $F$ and $T$}{The continuum limit: primary fields for Thompson's groups F and T}}\label{sec:primaryfields}
In this section we build observables on the semicontinuous limit Hilbert space $\mathcal{H}$ intended to represent \emph{smeared equal-time field operators} using the field-strength renormalization (\ref{eq:ffield}). We generalise everything to the case of nonregular partitions of $S^1$ as it costs no extra effort to do so (and is arguably more elegant). The discussions here are framed in terms of quantum spin systems, however, they generalise in a straightforward way to anyonic systems via trivalent categories.

Let $\mathcal{D}$ denote the directed set of standard dyadic partitions of $[0,1)$.
\begin{definition}
	Let $V\colon\mathbb{C}^d\to \mathbb{C}^d\otimes \mathbb{C}^d$, $\mathcal{E}(X) = V^\dag (X\otimes \mathbb{I}) V$, $\mathcal{E}(\mu^\alpha)= \lambda_\alpha\mu^\alpha$, and let $P\in\mathcal{D}$ be a standard dyadic partition. Let $f\in L^2([0,1],M_d(\mathbb{C}))$. Assume that $\lambda_\alpha\not=0$ for all $\alpha$ and define the following operator in $\mathcal{B}(\mathcal{H}_{P})$, where $\mathcal{H}_P=(\mathbb{C}^d)^{\otimes\abs{P}}$,
	\begin{equation}\label{eq:smearedfieldoperator}
		\phi_{{P}}(f) = \sum_{\substack{\alpha = 0 \\ \lambda_\alpha \not= 0}}^{d^2-1}\sum_{I\in {P}} \overline{f}_\alpha(I) (\lambda_\alpha)^{\log_2(|I|)}\mu_I^\alpha,
	\end{equation}
	where
	\begin{equation}
		\bar{f}_\alpha(I) = \frac{1}{d}\int_{I} \tr\bigl((\nu^\alpha)^\dag f(x)\bigr) dx
	\end{equation}
	and $\lambda_\alpha$ and $\nu^\alpha$ are the corresponding ascending weights and dual operators for $\mathcal{E}$.

\end{definition}
\begin{remark}
	The discretised field operator $\phi_{{P}}(f)$ is meant to represent a continuum field operator first smeared out by $f$ and then \emph{discretised}, \emph{averaged}, or \emph{coarse grained}, over the intervals making up the partition $P$. Intuitively, as the partition $P$ is taken finer and finer we should recover a dyadic version of the standard smeared field operator $\phi(f)$ of quantum field theory.
\end{remark}

We want to use the discretized field operator $\phi_P(f)$ to build $n$-point correlation functions. In the language of physics, this is achieved by replacing the smearing function $f$ with a delta function $f(x)=\delta(x-z)A$, where $A\in M_d(\mathbb{C})$. Mathematically this can be achieved by using operator-valued distributions. However, in order to avoid a long digression on distributions, we make the substitution $f(x)=\delta(x-z)M$ in (\ref{eq:smearedfieldoperator}) and note that
\begin{equation}
	\bar{f}_\alpha(I) = \tr\left((\nu^\alpha)^\dag M\right)\int \delta(x-z)\chi_I(x)\, dx = \tr\bigl((\nu^\alpha)^\dag M\bigr)\chi_I(z),
\end{equation}
where $\chi_I$ denotes the indicator function of $I$.
From this, we can distil the following \emph{ad hoc} definition.

\begin{definition}
	The \emph{discretised field operator} of type $\alpha$ at $z\in S^1$ with respect to the partition $P$ is defined to be
	\begin{equation}\label{eq:discretefieldoperator}
		\phi_{{P}}^\alpha(z) = \sum_{I\in {P}} \chi_I(z) (\lambda_\alpha)^{\log_2(|I|)}\mu_I^\alpha.
	\end{equation}
\end{definition}

Using products of discretised field operators at various positions allows us to define $n$-point functions. The general idea is that, for a given state $|\psi\rangle \in\mathcal{H}$, the $n$-point function should look like
\begin{equation}
	\text{``}C_{|\psi\rangle}^{\alpha_1\alpha_2\cdots\alpha_n}(x_1, x_2, \ldots, x_n) \sim \langle \psi|\hat{\phi}^{\alpha_1}(x_1)\hat{\phi}^{\alpha_2}(x_2)\cdots \hat{\phi}^{\alpha_n}(x_n)|\psi\rangle\text{''},
\end{equation}
with the subscript $P$ removed. Here the $\hat{\phi}^\alpha(x)$ are some putative \emph{quantum field operators}. To actually get this to work we need a limit to eliminate the partition $P$ in the definition (\ref{eq:discretefieldoperator}) of the discretised field operator. We achieve this by taking a limit over all partitions refining $P$. The existence of these limits is not obvious; we introduce some auxiliary definitions to facilitate the proof.

\begin{definition}
	Let $\mathbf{x} = (x_1, x_2, \ldots, x_n)$ be an ordered tuple of numbers lying in $[0,1)$, i.e., $0\le x_1 <x_2<\cdots < x_n < 1$. We say that $P\in\mathcal{D}$ is a \emph{supporting partition} for $\mathbf{x}$ if in any interval $I\in P$ there is \emph{at most} one of elements of the tuple $\mathbf{x}$ in $I$. We say that $P$ is a \emph{minimal supporting partition} for $\mathbf{x}$ if there is no coarser supporting partition.
\end{definition}
\begin{example}
	Let $\mathbf{x} = (\tfrac{1}{7}, \tfrac{2}{3}, \tfrac{5}{6})$. Then the partition
	\begin{equation}
		\{ [0,\tfrac14), [\tfrac14,\tfrac12), [\tfrac12,\tfrac34), [\tfrac34,1]\}
	\end{equation}
	is supporting for $\mathbf{x}$:
	\begin{center}
		\begin{tikzpicture}[scale=5]
			\draw (0, 0) -- (1,0);
			\draw (0, -0.025) -- (0, 0.025);
			\draw (0.25, -0.025) -- (0.25, 0.025);
			\draw (0.5, -0.025) -- (0.5, 0.025);
			\draw (0.75, -0.025) -- (0.75, 0.025);
			\draw (1, -0.025) -- (1, 0.025);
			\draw[fill=black] (1/7, 0) circle (0.01);
			\draw[fill=black] (2/3, 0) circle (0.01);
			\draw[fill=black] (5/6, 0) circle (0.01);
		\end{tikzpicture}
	\end{center}
	The partition $\{ [0,\tfrac12), [\tfrac12,\tfrac34), [\tfrac34,1]\}$ is a minimal supporting partition:
	\begin{center}
		\begin{tikzpicture}[scale=5]
			\draw (0, 0) -- (1,0);
			\draw (0, -0.025) -- (0, 0.025);
			\draw (0.5, -0.025) -- (0.5, 0.025);
			\draw (0.75, -0.025) -- (0.75, 0.025);
			\draw (1, -0.025) -- (1, 0.025);
			\draw[fill=black] (1/7, 0) circle (0.01);
			\draw[fill=black] (2/3, 0) circle (0.01);
			\draw[fill=black] (5/6, 0) circle (0.01);
		\end{tikzpicture}
	\end{center}
\end{example}
One can prove the following two lemmas.
\begin{lemma}
	Let $\mathbf{x}$ be an ordered tuple in $[0,1)$ and let $P$ support $\mathbf{x}$. Suppose that $Q$ refines $P$, i.e., $P\preceq Q$. Then $Q$ supports $\mathbf{x}$.
\end{lemma}
\begin{lemma}
	Let $\mathbf{x}$ be an ordered tuple lying in $[0,1)$. The minimal supporting partition for $\mathbf{x}$ is unique.
\end{lemma}
\begin{proof}
	Recall that $\mathcal{T}$ is the infinite binary tree of standard dyadic intervals whose nodes are labelled by $[\tfrac{a}{2^m}, \tfrac{a+1}{2^m})$, $a\in\{0,1,\ldots, 2^m-1\}$. Denote by $V$ the vertices of the tree (which are in bijection with the standard dyadic intervals) and by $I_v$ the standard dyadic interval associated to a vertex $v$ of the tree $\mathcal{T}$.

	Define the function $n:V\rightarrow \mathbb{Z}^+$ via
	\begin{equation}
		n(v) = |\{x_j\in\mathbf{x}\,|\, x_j\in I_v\}|.
	\end{equation}
	The function $n$ has the property that
	\begin{equation}\label{eq:nleafadds}
		n(v) = n(\text{left leaf of $v$}) + n(\text{right leaf of $v$}).
	\end{equation}
	Find the subtree $T_P = (V_P, E_P)$ of $\mathcal{T}$ defined by the property that $n(v)>1$ for all $v\in V_P$. This induces\footnote{Recall that a vertex-induced subgraph of a graph $G$ is a subset of the vertices along with all edges in $G$ whose endpoints lie in the subset.} a connected subtree by virtue of (\ref{eq:nleafadds}). Deleting $T_P$ (and all its associated edges in $E_P$) from $\mathcal{T}$ gives $m$ disconnected infinite binary trees whose root nodes induce a minimal supporting partition. Any minimal supporting partition would have to exclude $T_P$ and hence $P$ so constructed is unique.
\end{proof}
The utility of minimal supporting partitions for tuples $\mathbf{x}$ is that they directly allow us to reduce the computation of an $n$-point correlation function in the limit of fine partitions to a finite computation. To see this we specialise henceforth to the $n$-point functions of the \emph{vacuum vector} $|\Omega\rangle\in\mathcal{H}$.

Consider an $n$-tuple $\mathbf{x}$ in $[0,1)$ and let $P$ be its minimal supporting partition. Define for any tuple $\boldsymbol{\alpha} = (\alpha_1, \alpha_2, \ldots, \alpha_n)$ and any $Q\in\mathcal{D}$ refining $P$, i.e., $P\preceq Q$, the operator
\begin{equation}
	M^{\boldsymbol{\alpha}}_Q(\mathbf{x}) = \prod_{j=1}^n \phi_Q^{\alpha_j}(x_j).
\end{equation}
Since $Q$ is a supporting partition (it refines the minimal supporting partition) each of the factors in the product commutes with the others. Therefore the expression is well-defined.
\begin{lemma}
	Suppose $Q\in\mathcal{D}$ is a partition refining the minimal supporting partition $P$ of a tuple $\mathbf{x}$, i.e., $P\preceq Q$. Then
	\begin{equation}
		\langle \Omega_P|M^{\boldsymbol{\alpha}}_P(\mathbf{x})|\Omega_P\rangle = \langle \Omega_Q|M^{\boldsymbol{\alpha}}_Q(\mathbf{x})|\Omega_Q\rangle,
	\end{equation}
	where $[|\Omega_P\rangle] = [|\Omega_Q\rangle]$.
\end{lemma}
\begin{proof}
	The first observation we make is that $T^P_Q$ acts in a simple way on ascending operators localised to intervals $I\in Q$: let $f\in\text{Mor}(|P|,|Q|)$ be the planar forest connecting the objects $|P|$ and $|Q|$ corresponding to the isometry $T^P_Q$ and note
	\begin{equation}
		(T^P_Q)^\dag(\mu^\alpha_I)T^P_Q = (\lambda_\alpha)^{d_f(I,J)-1}\mu^\alpha_J,
	\end{equation}
	where $d_f(I,J)$ is the number of edges in the planar forest connecting the leaf node associated to $I$ to the node associated with its corresponding root $J$. If $P$ and $Q$ are supporting partitions for $\mathbf{x}$ then the intervals in $P$ (respectively, $Q$) containing the elements $x_j$ of the tuple $\mathbf{x}$ belong to disconnected components of the planar forest $f$.
	Denote these intervals by $I_j$ (i.e., one for each $x_j$) and their corresponding roots by $J_j$.	The ascending operators $\mu_I^\alpha$ and $\mu_J^\alpha$ before and after the action of $T^P_Q$ all commute and we have that
	\begin{equation}
		(T^P_Q)^\dag\biggl(\prod_j\mu^\alpha_{I_j}\biggr)T^P_Q = \prod_j(\lambda_\alpha)^{d_f(I_j,J_j)-1}\mu^\alpha_{J_j}.
	\end{equation}
	Noting that $(\lambda_\alpha)^{d_f(I,J)-1} = (\lambda_\alpha)^{\log_2(|J|)-\log_2(|I|)}$ and taking expectations gives us the result.
\end{proof}
We have now assembled enough information to prove the following
\begin{theorem}\label{thm:npt}
	Let $\mathbf{x}$ be an ordered $n$-tuple in $[0,1)$ and $\boldsymbol{\alpha}$ be an $n$-tuple in $\{0,1,\ldots,d^2-1\}^{\times n}$. Then the limit 
	\begin{equation}
		C^{\alpha_1\alpha_2\cdots\alpha_n}(x_1, x_2, \ldots, x_n) = \lim_{P\preceq Q} \langle \Omega_Q|M^{\boldsymbol{\alpha}}_Q(\mathbf{x})|\Omega_Q\rangle
	\end{equation}
	exists and is equal to
	\begin{equation}
		\langle \Omega_P|M^{\boldsymbol{\alpha}}_P(\mathbf{x})|\Omega_P\rangle,
	\end{equation}
	where $P$ is the minimal supporting partition of $\mathbf{x}$.
\end{theorem}
\begin{proof}
	Write $e_Q = \langle \Omega_Q|M^{\boldsymbol{\alpha}}_Q(\mathbf{x})|\Omega_Q\rangle$. We need to argue that the net $(e_Q)$ is eventually in any neighbourhood around $\langle \Omega_P|M^{\boldsymbol{\alpha}}_P(\mathbf{x})|\Omega_P\rangle$. But this is immediate since there always exists $R\in\mathcal{D}$ such that $P\preceq R$ and $Q\preceq R$: for any partition $S$ refining $R$ we have that $e_S = e_P$, i.e., $e_Q$ is eventually equal to $e_P$. Since $\mathbb{C}$ is Haussdorff the proof is complete.
\end{proof}

	This theorem tells us that an arbitrary $n$-point function is well-defined and, further, is computable in terms of operators on a finite-dimensional Hilbert space.

\begin{corollary}
	Let $f$ be an element of Thompson's group $F$ (respectively, $T$), and let $|f\rangle = U(f)|\Omega\rangle \in \mathcal{H}$ be the vector in the unitary representation afforded by $V$ resulting from applying $f$. Suppose $\mathbf{x}$ is an ordered $n$-tuple in $[0,1)$ and let $\boldsymbol{\alpha}$ be an $n$-tuple in $\{0,1,\ldots,d^2-1\}^{\times n}$. Then the limit
	\begin{equation}
		C^{\alpha_1\alpha_2\cdots\alpha_n}_{|f\rangle}(x_1, x_2, \ldots, x_n) = \lim_{R\preceq Q} \langle f_Q|M^{\boldsymbol{\alpha}}_Q(\mathbf{x})|f_Q\rangle
	\end{equation}
	exists and is equal to
	\begin{equation}
		\langle \Omega_{P'}|U(f)^\dag {T^{f(P')}_R}^\dag M^{\boldsymbol{\alpha}}_R(\mathbf{x})T^{f(P')}_R U(f)|\Omega_{P'}\rangle,
	\end{equation}
	where $P$ is the minimal supporting partition of $\mathbf{x}$, $P'\succeq P$ is good for $f$, and $R$ refines both $P$ and $f(P')$. (Here we use the notation $|f_Q\rangle$ for a representation of $|f\rangle$ on partition $Q$.)
\end{corollary}

\begin{remark}
	There is more work to do if we want to realise the $n$-point functions as expectation values
\begin{equation}\label{eq:npointcorrfn}
C^{\alpha_1\alpha_2\cdots\alpha_n}(x_1, x_2, \ldots, x_n) = \langle \Omega|\hat{\phi}^{\alpha_1}(x_1)\hat{\phi}^{\alpha_2}(x_2)\cdots \hat{\phi}^{\alpha_n}(x_n)|\Omega\rangle
\end{equation}
of genuine quantum field operators $\hat{\phi}^{\alpha}(x)$. This would require us to show that the field operators obey the usual mathematical properties required of a quantum field operator, namely,
that $\hat{\phi}^{\alpha}(x)$ is an operator-valued distribution such that there is a dense subspace $D\subset \mathcal{H}$ of our Hilbert space such that
\begin{enumerate}
	\item For each Schwarz function $f\in\mathcal{S}$ the domain of definition $D(\hat{\phi}^{\alpha}(f))$ contains $D$.
	\item The induced map $\mathcal{S}\mapsto \text{End}(D)$ via $f\mapsto \hat{\phi}^{\alpha}(f)$ is linear.
	\item For every $|\psi\rangle \in D$ and $|\phi\rangle \in \mathcal{H}$, the inner product
	\begin{equation}
		\langle \phi|\hat{\phi}^{\alpha}(f)|\psi\rangle
	\end{equation}
	is a tempered distribution.
\end{enumerate}
To carry this out there seems to be no way around a Wightman-type reconstruction argument. We avoid this here, restricting our attention to simply analysing the properties of the $n$-point functions.
\end{remark}

\section{\texorpdfstring{Short-distance behaviour of the $n$-point functions}{Short-distance behaviour of the n-point functions}}\label{sec:shortdistance}
Many of the properties of $C^{\alpha_1\alpha_2\cdots\alpha_n}(x_1, x_2, \ldots, x_n)$ are immediate consequences of the formula in Theorem~\ref{thm:npt}. The first important result concerns the short-distance behaviour of the two-point function $C^{\alpha\beta}(x,y)$. To understand this we focus first on the case where $x$ and $y$ are dyadic, in which case we can express them in binary as
\begin{equation}
	\begin{split}
		x &= 0.x_{-1}x_{-2}\cdots x_{-m}, \quad\text{and} \\
		y &= 0.y_{-1}y_{-2}\cdots y_{-n},
	\end{split}
\end{equation}
where $x_j\in\{0,1\}$ and $y_j \in \{0,1\}$. (Such expansions are assumed to have an infinite sequence of trailing zeroes.) Without loss of generality we assume that $n>m$. The minimal supporting partition for the pair $(x,y)$ is easy to derive: first express $x = \overline{x} + x'$ and $y = \overline{x} + y'$, where $\overline{x}$ contains the first $l$ digits of the binary expansions of $x$ and $y$ which are in common. Now recursively subdivide the interval $[0,1]$ according to the following recipe: set $I\leftarrow [0,1]$ and $j\leftarrow -1$ and repeat steps (1) and (2) while $j \ge -l-1$:
\begin{enumerate}
	\item subdivide $I$ into $I=I_0\cup I_1$ and set $I\leftarrow I_{\overline{x}_{j}}$, where $\overline{x}_{j}$ is the $j$th digit of $\overline{x}$,
	\item set $j\leftarrow j-1$.
\end{enumerate}
The subdivisions carried out via this procedure induce a standard dyadic partition $P$ which is minimal for the pair $(x,y)$. Note that at the final iteration $x$ and $y$ are located in neighbouring intervals of length $2^{-l-1}$. Indeed, the two intervals separating $x$ and $y$ are none other than $I = [\overline{x},\overline{x}+\frac{1}{2^{l+1}})$ and $I'=[\overline{x}+\frac{1}{2^{l+1}}, \overline{x}+\frac{1}{2^{l}})$.
Now that we have the minimal separating partition we can immediately apply Theorem~\ref{thm:npt} to deduce the two-point function
\begin{equation}
	C^{\alpha\beta}(x,y) = \langle \Omega_P|(\lambda_\alpha^{-l-1}\mu^\alpha_I)(\lambda_\beta^{-l-1}\mu^\beta_{I'})|\Omega_P\rangle
\end{equation}
By making use of the structure constants for $\star$ we can explicitly evaluate this expression. This is summarised in the following lemma.
\begin{lemma}
	Let $0\le x<y <1$ be two dyadic fractions and let $\alpha,\beta \in {0,1,\ldots, d^2-1}$. Write $x = \overline{x} + x'$ and $y = \overline{x} + y'$, where $\overline{x}$ contains the first $l$ digits of the binary expansions of $x$ and $y$ which are in common. Then
	\begin{equation}
		C^{\alpha\beta}(x,y) = \sum_{\gamma=0}^{d^2-1}\lambda_\gamma^{-1} D(x,y)^{\log_2(\lambda_\alpha)+\log_2(\lambda_\beta)-\log_2(\lambda_\gamma)}{f^{\alpha\beta}}_{\gamma}\langle\Omega_{[0,1]}|\mu^\gamma|\Omega_{[0,1]}\rangle,
	\end{equation}
	where $D(x,y) = 2^{-l-1}$ is the \emph{coarse-graining distance} between $x$ and $y$.
\end{lemma}
\begin{proof}
	Start with the expression
	\begin{equation}
		C^{\alpha\beta}(x,y) = \langle \Omega_P|(\lambda_\alpha^{-l-1}\mu^\alpha_I)(\lambda_\beta^{-l-1}\mu^\beta_{I'})|\Omega_P\rangle.
	\end{equation}
	Since the intervals $I = [\overline{x},\overline{x}+\frac{1}{2^{l+1}})$ and $I'=[\overline{x}+\frac{1}{2^{l+1}}, \overline{x}+\frac{1}{2^{l}})$ are neighbours we can exploit the $\star$ operation to evaluate this expression on the coarse-grained partition $P'$ where the neighbouring intervals $I$ and $I'$ are joined to the interval $I_{\overline{x}}$ of length $l$:
	\begin{equation}
		C^{\alpha\beta}(x,y) = \sum_{\gamma=0}^{d^2-1}(\lambda_\alpha\lambda_\beta)^{-l-1} {f^{\alpha\beta}}_{\gamma}\langle\Omega_{P'}|\mu^\gamma_{I_{\overline{x}}}|\Omega_{P'}\rangle.
	\end{equation}
	This expression is easy to simplify via the action of the CP map $\mathcal{E}$:
	\begin{equation}
		C^{\alpha\beta}(x,y) = \sum_{\gamma=0}^{d^2-1}{f^{\alpha\beta}}_{\gamma}(\lambda_\alpha\lambda_\beta)^{-l-1}\lambda_\gamma^l \langle\Omega_{[0,1]}|\mu^\gamma|\Omega_{[0,1]}\rangle.
	\end{equation}
	Now write $l+1=-\log_2(D(x,y))$: we finally obtain
	\begin{equation}
		C^{\alpha\beta}(x,y) = \sum_{\gamma=0}^{d^2-1}\lambda_\gamma^{-1} D(x,y)^{\log_2(\lambda_\alpha)+\log_2(\lambda_\beta)-\log_2(\lambda_\gamma)}{f^{\alpha\beta}}_{\gamma}\langle\Omega_{[0,1]}|\mu^\gamma|\Omega_{[0,1]}\rangle.\qedhere
	\end{equation}
\end{proof}
\begin{remark}
	When $x$ and $y$ are a \emph{standard dyadic pair}, that is, $x = \frac{a}{2^l}$ and $y = \frac{a+1}{2^l}$, with $l\in\mathbb{Z}^+$ and $a \in \{0,1,\ldots, 2^l-1\}$, then $D(x,y) = |x-y|$, so that we can rewrite
	\begin{equation}\label{eq:2ptdyadicpair}
		C^{\alpha\beta}(x,y) = \sum_{\gamma=0}^{d^2-1}\lambda_\gamma^{-1} |x-y|^{\log_2(\lambda_\alpha)+\log_2(\lambda_\beta)-\log_2(\lambda_\gamma)}{f^{\alpha\beta}}_{\gamma}\langle\Omega_{[0,1]}|\mu^\gamma|\Omega_{[0,1]}\rangle.
	\end{equation}
\end{remark}

In the context of conformal field theory an expression such as (\ref{eq:2ptdyadicpair}) for standard dyadic pairs is especially suggestive. We therefore propose the following prototype definition for the analogue of the scaling dimension.
\begin{definition}
	The \emph{scaling dimension} $h_\alpha$ for the field $\hat{\phi}^\alpha(x)$ is
	\begin{equation}
		h_\alpha = -\operatorname{Re}\log_2(\lambda_\alpha).
	\end{equation}
\end{definition}

Contrary to the situation in conformal field theory there is no reason to expect that, in general,
\begin{equation}\label{eq:2pointCFTcorr}
	C^{\alpha\beta}(x,y) \sim C^{\alpha\beta} D(x,y)^{-2h}
\end{equation}
only when $h_\alpha = h = h_\beta$. We hence promote (\ref{eq:2pointCFTcorr}) to a \emph{necessary condition} for the existence of a physical continuum limit.

Let's take a look at the two-point correlation function for some prototypical examples. In general the correlation function is
\emph{discontinuous}. It typically behaves something like $C^{\alpha\beta}(x,y)\sim \left\lceil\right|x-y|^{-h}\rceil$, however, caution must be taken as the behaviour of the correlation function is asymmetric about $|x-y|=0$. Here we have illustrated the example $x=0$ in the case where $-\log_2(\lambda_\alpha) = \frac{1}{\sqrt{2}} = \log_2(\lambda_\beta)$:
\begin{center}
	\includegraphics{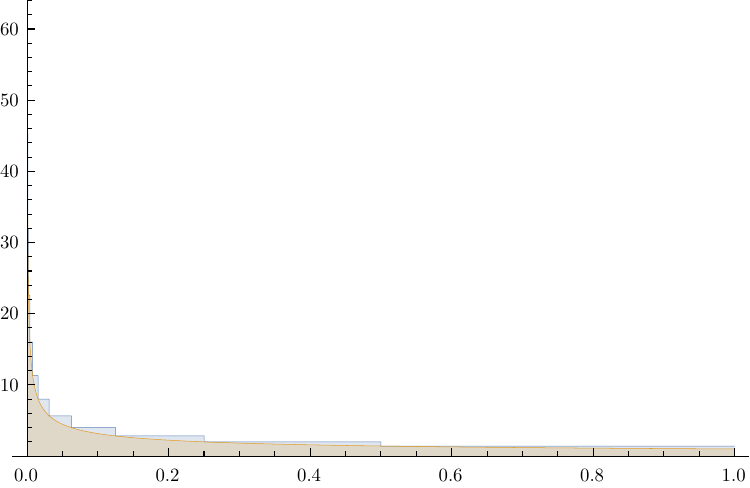}
\end{center}
The blue staircase-like line is the the correlation function itself. The brown line is the envelope given by $|x-y|^{-\frac{1}{2}}$.
In the second figure below we see the example where $x = 5/8$:
\begin{center}
	\includegraphics{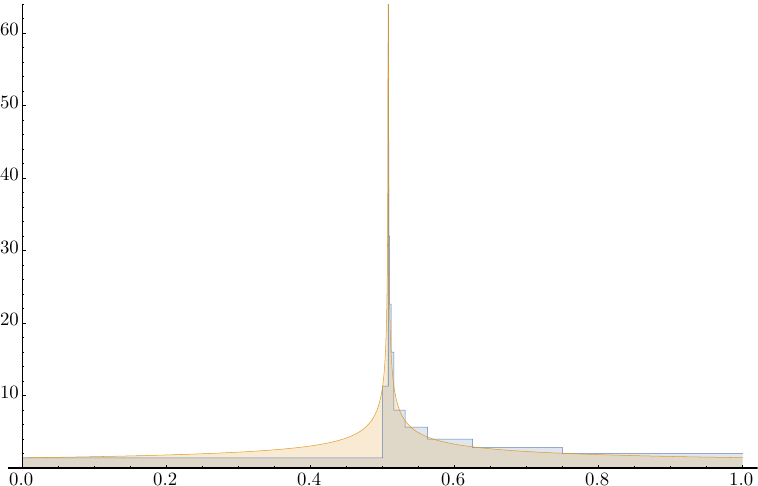}
\end{center}
Notice that the behaviour to the left of the point $x = 5/8$ is different to that on the right.

These two examples highlight the important fact that the continuum correlation functions for a tree state may be discontinuous and asymmetric.

\section{Fusion rules and the operator product expansion}\label{sec:ope}
So far we have studied the two-point correlation function. Now we look at the three point function in an attempt to obtain an analogue of the \emph{operator product expansion}.

In general, a three point function is of the form
\begin{equation}
	C^{\alpha\beta\gamma}(x,y,z) = \langle\Omega|\hat{\phi}^{\alpha}(x)\hat{\phi}^{\beta}(y)\hat{\phi}^{\gamma}(z)|\Omega\rangle.
\end{equation}
for $x,y,z\in[0,1)$ with $x<y<z$.

We can compute this correlation function by first finding the minimal supporting partition $P$ for $(x,y,z)$ and setting
\begin{equation}
	C^{\alpha\beta\gamma}(x,y,z) = \langle\Omega_P|(\lambda_{\alpha}^{-l-1}\mu^{\alpha}_I)(\lambda_{\beta}^{-m-1}\mu^{\beta}_J)(\lambda_{\gamma}^{-n-1}\mu^{\gamma}_K)|\Omega_P\rangle,
\end{equation}
where $I$, $J$, and $K$ are the intervals containing $x$, $y$, and $z$, respectively.

To calculate this expression note that we can exploit the formulas we already have for the two-point function. The important observation here is that when $d_T(x,y) < d_T(y,z)$ we can first fuse operators $\mu^{\alpha}$ and $\mu^{\beta}$ resulting in some linear combination of $\mu^{\gamma'}s$ and then we fuse these with $\mu^\gamma$. Correspondingly, if $d_T(x,y) > d_T(y,z)$ we first fuse the last two then fuse on the first operator. Thus, the three-point function is completely determined by knowledge of the fusion coefficients ${f^{\alpha\beta}}_\gamma$. The observation is also particularly reminiscent of the \emph{operator product expansion} (OPE). We exemplify this by promoting it to a prototype definition.
\begin{definition}
	Given formal \emph{primary fields} $\phi^\alpha(x)$, $\alpha = 0, 1, \ldots, d^2-1$, the formal short-distance expansion
	\begin{equation}
		  \hat{\phi}^\alpha(x)\hat{\phi}^\beta(y) \sim \sum_{\gamma=0}^{d^2-1} {f^{\alpha\beta}}_{\gamma}D(x,y)^{h_\gamma-h_\alpha-h_\beta} \hat{\phi}^\gamma(y)
	\end{equation}
	is called the \emph{operator product expansion}.
\end{definition}

Here the $\sim$ means that the expression only makes sense in a correlation function, and that oscillatory behaviour is neglected, that is, we only study the divergence up to an overall absolute value sign.

A crucial role is played by the structure of the dimensions $h_\alpha$ as they control, via the quantity $h_\gamma-h_\alpha-h_\beta$, the \emph{divergence} of the $n$-point correlation functions as $x\rightarrow y$.

The fusion coefficients ${f^{\alpha\beta}}_{\gamma}$ determine the structure of the three-point function. In particular, whether ${f^{\alpha\beta}}_{\gamma}=0$ or not determines whether a given correlation function is nontrivial or not. To this end we introduce the following three-index tensor
\begin{equation}
	{N^{\alpha\beta}}_\gamma = \begin{cases}
		1, \quad \text{if ${f^{\alpha\beta}}_{\gamma}\not=0$ and}\\ 0,\quad \text{otherwise.}
	\end{cases}
\end{equation}
This tensor can be used to construct an (in general) nonassociative and noncommutative algebra $\mathcal{A}$ over $\mathbb{Z}$. As a set we define $\mathcal{A}$ to be the lattice
\begin{equation}
	\mathcal{A} = \mathbb{Z}^{d^2},
\end{equation}
and we choose some basis $\{\phi^\alpha\}_{\alpha\in I}$, $I=\{0,1,\ldots,d^2-1\}$, and introduce the product operator $\star$ via
\begin{equation}
	\phi^\alpha\star \phi^\beta = \sum_{\gamma\in I} {N^{\alpha\beta}}_\gamma \phi^\gamma.
\end{equation}

Usually the algebra $\mathcal{A}$ will be neither associative nor commutative. However, in special cases, it can be the case that ${N^{\alpha\beta}}_\gamma$ ends up satisfying these additional constraints. In this case $\mathcal{A}$ becomes a \emph{fusion ring}. We can obtain a representation for the fusion ring via the commuting matrices $N^\alpha$ with matrix elements
\begin{equation}
	[N^\alpha]_{\beta\gamma} = {N^{\alpha\beta}}_\gamma.
\end{equation}

\section{\texorpdfstring{The action of Thompson's groups $F$ and $T$ on $n$-point functions}{The action of Thompson's groups F and T on n-point functions}}\label{sec:thompsonaction}
The analogy between CFT and quantum mechanics symmetric under Thompson's groups $F$ and $T$ manifests itself most strongly when considering how $n$-point
functions transform under Thompson group elements $f$. Here we discuss the $n$-point correlation function with respect to the vacuum vector $|\Omega\rangle$ and its transformed version $U(f)|\Omega\rangle$.

\begin{theorem}\label{thm:npointaction}
Let $f\in T$ be an element of Thompson's group $T$ and $U(f)$ its unitary representation. Then the action of $T$ on $\mathcal{H}$ in terms of $n$-point functions is
	\begin{equation}\label{eq:nptaction}
		C^{\boldsymbol{\alpha}}(x_1, x_2, \ldots, x_n) = \prod_{j=1}^n \left(\frac{df}{dx}\bigg|_{x = x_j}\right)^{-h_{\alpha_j}}C^{\boldsymbol{\alpha}}_{|f\rangle}(f(x_1), f(x_2), \ldots, f(x_n)),
	\end{equation}
	where the limit in the derivative is taken above via $x\rightarrow x_j+\epsilon$.
\end{theorem}
\begin{proof}
	Let $P$ be the minimal supporting partition for the tuple $(x_1, x_2, \ldots, x_n)$. Choose a refinement $P'$ which is good for $f$ and choose $R$ which refines both $P'$ and $f(P')$. Then
	\begin{equation}
		(T^{P'}_R)^\dag M^{\boldsymbol{\alpha}}_{R}(\mathbf{x})T^{P'}_R = M^{\boldsymbol{\alpha}}_{P'}(\mathbf{x})
	\end{equation}
	and the LHS of (\ref{eq:nptaction}) is directly equal to
	\begin{equation}
		C^{\boldsymbol{\alpha}}(x_1, x_2, \ldots, x_n) = \langle\Omega_{P'}| M^{\boldsymbol{\alpha}}_{P'}(\mathbf{x}) | \Omega_{P'}\rangle = \langle\Omega_{R}| M^{\boldsymbol{\alpha}}_{R}(\mathbf{x}) | \Omega_{R}\rangle.
	\end{equation}
	Now we compare left and right-hand sides: the correlation function on the RHS of (\ref{eq:nptaction}) is the expectation value of
	\begin{equation}
		M^{\boldsymbol{\alpha}}_{f(P')}(f(x_1), f(x_2), \ldots, f(x_n)) = \prod_{j=1}^n (\lambda_{\alpha_j})^{\log_2(|f(I_j)|)}\mu_{f(I_j)}^{\alpha_j},
	\end{equation}
	with respect to $U(f)|\Omega_{P'}\rangle$, noting that $f(P')$ is a supporting partition for
	\begin{equation}
		(f(x_1), f(x_2), \ldots, f(x_n)).
	\end{equation}
	Rewriting this expression as
	\begin{equation}
		M^{\boldsymbol{\alpha}}_{f(P')}(f(x_1), f(x_2), \ldots, f(x_n)) = \prod_{j=1}^n (\lambda_{\alpha_j})^{\log_2(|f(I_j)|)-\log_2(|I_j|)}(\lambda_{\alpha_j})^{\log_2(|I_j|)}\mu_{f(I_j)}^{\alpha_j}
	\end{equation}
	and taking the expectation value with respect to $U(f)|\Omega_{P'}\rangle$ gives us
	\begin{equation}
		\text{RHS} = \prod_{j=1}^n (\lambda_{\alpha_j})^{\log_2(|f(I_j)|)-\log_2(|I_j|)} C^{\boldsymbol{\alpha}}_{|f\rangle}(f(x_1), f(x_2), \ldots, f(x_n)).
	\end{equation}
	Now the we can calculate the length of the interval $f(I_j)$ as follows
	\begin{equation}
		|f(I_j)| = \left(\frac{df}{dx} \bigg|_{x = x_j}\right) |I_j|.
	\end{equation}
	(Here the derivative is defined with a limit from the right so as to avoid singularities when $x_j$ is at a breakpoint.)
	Taking logs and exchanging exponents using the identity $a^{\log(b)} = b^{\log(a)}$ gives us the result.
\end{proof}

By substituting $x_j = f^{-1}(z_j)$ we can rewrite this result in a somewhat more useful form,
	\begin{equation}\label{eq:nptaction2}
		C^{\boldsymbol{\alpha}}_{|f\rangle}(z_1, z_2, \ldots, z_n) = \prod_{j=1}^n \left(\frac{df}{dz}\bigg|_{z = f^{-1}(z_j)}\right)^{h_{\alpha_j}} C^{\boldsymbol{\alpha}}(f^{-1}(z_1), f^{-1}(z_2), \ldots, f^{-1}(z_n)).
	\end{equation}
\begin{remark}
	This corollary tells us that knowledge of the ``vacuum'' $n$-point function
	\begin{equation}
		\langle \Omega|\hat{\phi}^{\alpha_1}(z_1)\hat{\phi}^{\alpha_2}(z_2)\cdots \hat{\phi}^{\alpha_n}(z_n)|\Omega\rangle
	\end{equation}
	alone is enough to calculate the $n$-point functions with respect to any transformed state $|f\rangle = U(f)|\Omega\rangle$:
	\begin{multline}
		\langle f|\hat{\phi}^{\alpha_1}(z_1)\hat{\phi}^{\alpha_2}(z_2)\cdots \hat{\phi}^{\alpha_n}(z_n)|f\rangle = \\ \prod_{j=1}^n \left(\frac{df}{dz}\bigg|_{z = f^{-1}(z_j)}\right)^{h_{\alpha_j}}\langle \Omega|\hat{\phi}^{\alpha_1}(f^{-1}(z_1))\hat{\phi}^{\alpha_2}(f^{-1}(z_2))\cdots \hat{\phi}^{\alpha_n}(f^{-1}(z_n))|\Omega\rangle.
	\end{multline}
\end{remark}

In the case where our unitary representation is determined by a planar perfect tangle we deduce that the correlation function is $\mathit{PSL}(2,\mathbb{Z})$-invariant because $|f\rangle = |\Omega\rangle$ for $f\in\mathit{PSL}(2,\mathbb{Z})$.

\subsection{The connection between Thompson group actions and smearing}
We have seen in this section how to calculate the $n$-point correlation function with respect to a non-vacuum state $|f\rangle$ prepared by applying a Thompson group transformation. Here we exploit this observation to calculate the expectation values of smeared field operators. Note that this treatment is not entirely rigorous.

Let's suppose we have a field operator $\hat{\phi}^\alpha(x)$, obtained via the limiting procedure described above. If we ``take away the expectation values'' in (\ref{eq:nptaction2}) we obtain the following transformation law for the field:
\begin{equation}
      U(f)^\dag \hat{\phi}^\alpha(z) U(f) = \left(\frac{df}{dz}\bigg|_{f^{-1}(z)}\right)^{h_{\alpha}} \hat{\phi}^\alpha(f^{-1}(z)).
\end{equation}
This expression is understood to make sense only in the expectation values.

We can now use this field operator to build the smeared operator
\begin{equation}
      \hat{\phi}^\alpha(\chi) = \int_{S^1} \chi(z) \hat{\phi}^\alpha(z)\, dz,
\end{equation}
where $\chi(x)$ is a, say, $L^1(S^1)$ function. The action of $T$ on smeared field operators is then given by
\begin{equation}
      U(f)^\dag \hat{\phi}^\alpha(\chi) U(f) = \int_{S^1} \chi(z) \left(\frac{df}{dz}\bigg|_{f^{-1}(z)}\right)^{h_{\alpha}} \hat{\phi}^\alpha(f^{-1}(z))\, dz
\end{equation}
Making the change of variable $x= f^{-1}(z)$, we find that
\begin{equation}
      U(f)^\dag \hat{\phi}^\alpha(\chi) U(f) = \int_{S^1} \chi(f(x)) \left(\frac{df}{dx}\right)^{h_{\alpha}+1} \hat{\phi}^\alpha(x)\, dx.
\end{equation}
When $\chi(z)$ is the constant function $\chi(x)=1$ we see that $U(f)^\dag \hat{\phi}^\alpha(\chi) U(f)$ is the field operator smeared by a simple function
\begin{equation}
      U(f)^\dag \hat{\phi}^\alpha(\chi) U(f) = \sum_{I\in \mathcal{P}} 2^{c_I(h_\alpha+1)}\hat{\phi}^\alpha(\chi_I) = \hat{\phi}^\alpha(s),
\end{equation}
where
\begin{equation}
      s(x) = \sum_{I\in \mathcal{P}}2^{c_I(h_\alpha+1)}\chi_I(x),
\end{equation}
$\chi_I$ is the indicator function for the set $I$, $\mathcal{P}$ is a partition good for $f$, and $2^{c_I}$ is the gradient of the function $f$ on the interval $I$.

What we see is that expectation values with respect to the vacuum of a field operator smeared with certain special simple functions can be directly related to the expectation values of the evenly smeared field operator with respect to a transformed vacuum.

\section{Application: Spin system}\label{sec:example1}

Here we apply the formalism of the previous sections to a simple example quantum spin system comprised of a lattice of \emph{qutrits}, i.e.,
\begin{equation}
	\mathcal{H}_N \cong \bigotimes_{j=0}^{2^m-1} \mathbb{C}^3,
\end{equation}
where, as usual, we set $N=2^m$. We employ quantum notation and choose the perfect tensor $V\colon\mathbb{C}^3\to\mathbb{C}^3\otimes\mathbb{C}^3$ from eq.~(\ref{eq:perfect-example}) given by
\begin{equation}
	\langle jk|V|l\rangle = \begin{cases}
		0 & \text{if $j=k$, $k=l$, or $l=j$,}\\
		\frac{1}{\sqrt{2}} & \text{otherwise.}
	\end{cases}
\end{equation}
The ascending operator $\mathcal{E}$ constructed from this perfect tensor has the three eigenvalues
\begin{equation}
	\lambda_1=1, \qquad \lambda_\alpha=-\frac{1}{2}, \qquad \lambda_\beta=\frac{1}{2}.
\end{equation}
$\lambda_1=1$ has the (right) eigenvector $\mu^1=\mathbb{I}$; $\lambda_\alpha=-\frac{1}{2}$ has eigenvectors
\begin{gather}
	\mu^{\delta^1}=\begin{pmatrix}
		-1&0&0\\
		0&0&0\\
		0&0&1
	\end{pmatrix}, \qquad
	\mu^{\alpha^1}=\begin{pmatrix}
		0&0&0\\
		0&0&-1\\
		0&1&0
	\end{pmatrix}, \qquad
	\mu^{\alpha^2}=\begin{pmatrix}
		0&0&-1\\
		0&0&0\\
		1&0&0
	\end{pmatrix},\\
	\mu^{\delta_2}=\begin{pmatrix}
		-1&0&0\\
		0&1&0\\
		0&0&0
	\end{pmatrix}, \qquad
	\mu^{\alpha^3}=\begin{pmatrix}
		0&-1&0\\
		1&0&0\\
		0&0&0
	\end{pmatrix};
\end{gather}
$\lambda_{\beta}=\frac{1}{2}$ has eigenvectors
\begin{equation}
	\mu^{\beta^1}=\begin{pmatrix}
		0&0&0\\
		0&0&1\\
		0&1&0
	\end{pmatrix}, \qquad
	\mu^{\beta^2}=\begin{pmatrix}
		0&0&1\\
		0&0&0\\
		1&0&0
	\end{pmatrix}, \qquad
	\mu^{\beta^3}=\begin{pmatrix}
		0&1&0\\
		1&0&0\\
		0&0&0
	\end{pmatrix}.
\end{equation}
They result in the fusion rules
\begin{center}
	\begin{tabular}{|Sc|Sc|Sc|Sc|Sc|Sc|Sc|Sc|Sc|Sc|} \hline
		\rowcolor{gray}${\times}$&$1$&$\delta^1$&$\delta^2$&$\beta^1$&$\beta^2$&$\beta^3$&$\alpha^1$&$\alpha^2$&$\alpha^3$\\ \hline
	 	\cellcolor{gray}$1$&$1$&$\delta^1$&$\delta^2$&$\beta^1$&$\beta^2$&$\beta^3$&$\alpha^1$&$\alpha^2$&$\alpha^3$\\ \hline
		\cellcolor{gray}$\delta^1$&$\delta^1$&$\Sigma$&$\Sigma$&$\beta^1$&$0$&$\beta^3$&$\alpha^1$&$0$&$\alpha^3$\\ \hline
		\cellcolor{gray}$\delta^2$&$\delta^2$&$\Sigma$&$\Sigma$&$\beta^1$&$\beta^2$&$0$&$\alpha^1$&$\alpha^2$&$0$\\ \hline
		\cellcolor{gray}$\beta^1$&$\beta^1$&$\beta^1$&$\beta^1$&$\Sigma$&$\beta^3$&$\beta^2$&$0$&$\alpha^3$&$\alpha^2$\\ \hline
		\cellcolor{gray}$\beta^2$&$\beta^2$&$0$&$\beta^2$&$\beta^3$&$\Sigma$&$\beta^1$&$\alpha^3$&$0$&$\alpha^1$\\ \hline
		\cellcolor{gray}$\beta^3$&$\beta^3$&$\beta^3$&$0$&$\beta^2$&$\beta^1$&$\Sigma$&$\alpha^2$&$\alpha^1$&$0$\\ \hline
		\cellcolor{gray}$\alpha^1$&$\alpha^1$&$\alpha^1$&$\alpha^1$&$0$&$\alpha^3$&$\alpha^2$&$\Sigma$&$\beta^3$&$\beta^2$\\ \hline
		\cellcolor{gray}$\alpha^2$&$\alpha^2$&$0$&$\alpha^2$&$\alpha^3$&$0$&$\alpha^1$&$\beta^3$&$\Sigma$&$\beta^1$\\ \hline
		\cellcolor{gray}$\alpha^3$&$\alpha^3$&$\alpha^3$&$0$&$\alpha^2$&$\alpha^1$&$0$&$\beta^2$&$\beta^1$&$\Sigma$\\ \hline
	\end{tabular}
\end{center}
with $\Sigma=1+\delta^1+\delta^2$.

From the eigenvalues we get $h_1=0$ and $h_\alpha=h_\beta=1$. For the OPE, we give the two examples
\begin{align}
	\hat{\phi}^{\delta^1}(x)\hat{\phi}^{\delta^2}(y) &\sim -\frac{1}{6} D(x, y)^{-2} \hat{\phi}^1(y) -\frac{1}{3} D(x, y)^{-1} \hat{\phi}^{\delta^1}(y) - \frac{1}{3} D(x, y)^{-1}\hat{\phi}^{\delta^2}(y),\\
	\hat{\phi}^{\beta^2}(x)\hat{\phi}^{\alpha^3}(y) &\sim \frac{1}{3} D(x, y)^{-1} \hat{\phi}^{\alpha^1}(y).
\end{align}

\section{Application: the Fibonacci lattice}\label{sec:example2}

We now illustrate the formalism developed in the previous sections in terms of a tree state defined for the Fibonacci category $\mathcal{F}$. The computations in this section were performed using the \textit{TriCats} package \cites{stiegemannSupp,stiegemannTriCats}.

The fusion ring of $\mathcal{F}$ is generated by the two elements $1$ and $\tau$ and fusion rules
\begin{align}
	1 \times 1&= 1\\
	1 \times \tau &= \tau\\
	\tau\times\tau&=1+\tau.
\end{align}
$\mathcal{F}$ is a trivalent category with $\dim\mathcal{C}_4=2$ and $d=\frac{1}{2}(1\pm\sqrt{5})$, and it is a special case of an $\mathit{SO}(3)_q$ category with $q=4$. $\mathcal{C}_4$ is spanned by the two vectors
\begin{equation}\label{eq:fibbasis}
	\begin{tikzpicture}[scale=0.5]
    \draw (0, 0) .. controls (1, 1) .. (0, 2);
    \draw (2, 0) .. controls (1, 1) .. (2, 2);
  \end{tikzpicture}\,,\quad%
	\begin{tikzpicture}[scale=0.5]
    \draw (0, 0) .. controls (1, 1) .. (2, 0);
    \draw (0, 2) .. controls (1, 1) .. (2, 2);
  \end{tikzpicture}\,.
\end{equation}

We use a modification of the trivalent vertex, effectively doubling lines and replacing the trivalent vertex with
\begin{equation}
	V=\begin{tikzpicture}[scale=0.6,baseline=10mm]
    \draw[rounded corners] (2.5, 3.5) -- (2.5, 2) -- (1.875, 1.375);
    \draw (2.25, 1) -- (2.75, 0.5);
    \draw[rounded corners] (2.5, 0.75) -- (1.25, 2) -- (-0.25, 0.5);
    \draw[draw=white,double=black,double distance=0.4pt,line width=3pt] (2, 3.5) -- (2, 2) -- (0.5, 0.5);
		\draw (2, 2) -- (2.2, 1.8);
    \draw[draw=white,double=black,double distance=0.4pt,line width=3pt] (2.1, 1.9) -- (3.5, 0.5);
  \end{tikzpicture}.
\end{equation}
The braiding appearing in $V$ is given by
\begin{equation}
	\begin{tikzpicture}[scale=0.5,baseline=3.5mm]
    \draw (0, 2) -- (2, 0);
    \draw[draw=white,double=black,double distance=0.4pt,line width=3pt] (0, 0) -- (2, 2);
  \end{tikzpicture}
	=
	\begin{tikzpicture}[scale=0.5,baseline=3.5mm]
    \draw (0, 0) .. controls (1, 1) .. (0, 2);
    \draw (2, 0) .. controls (1, 1) .. (2, 2);
  \end{tikzpicture}
	+e^{4i\pi/5}
  \begin{tikzpicture}[scale=0.5,baseline=3.5mm]
    \draw (0, 0) .. controls (1, 1) .. (2, 0);
    \draw (0, 2) .. controls (1, 1) .. (2, 2);
  \end{tikzpicture}.
\end{equation}
In the basis (\ref{eq:fibbasis}), the ascending operator has matrix elements
\begin{equation}
	\begin{pmatrix}
		1 & \frac{1}{2}(3-\sqrt{5})\\
		0 & \frac{1}{2}(3-\sqrt{5})
	\end{pmatrix}
\end{equation}
and eigenvalues
\begin{equation}
	\lambda_1=1, \qquad \lambda_\tau=\frac{1}{2}\bigl(3-\sqrt{5}\bigr).
\end{equation}
The fusion coefficients are given by
\begin{equation}
	f^1=\begin{pmatrix}
		1 & 0\\
		0 & \frac{1}{2}(3-\sqrt{5})
	\end{pmatrix},\quad
	f^\tau=\begin{pmatrix}
		0 & \frac{1}{2}(3-\sqrt{5})\\
		\sqrt{5}-2 & 5-2\sqrt{5}
	\end{pmatrix}.
\end{equation}
The OPE then gives us the short-distance behaviour
\begin{align}
	\hat{\phi}^1(x)\hat{\phi}^1(y) &\sim \hat{\phi}^1(y), \\
	\hat{\phi}^1(x)\hat{\phi}^\tau(y) &\sim \frac{1}{2}(3-\sqrt{5}) \hat{\phi}^\tau(y), \\
	\hat{\phi}^\tau(x)\hat{\phi}^\tau(y) &\sim (\sqrt{5}-2)D(x, y)^{-2h_\tau}	\hat{\phi}^1(y) + (5-2\sqrt{5})D(x, y)^{-h_\tau}\hat{\phi}^\tau(y)
\end{align}
with $h_\tau=-\log_2\bigl( \frac{1}{2}(3-\sqrt{5}) \bigr)\approx 1.388$.
We obtain a representation of the fusion ring via the matrices
\begin{equation}
	N^1=\begin{pmatrix}
		1 & 0\\
		0 & 1
	\end{pmatrix},\quad
	N^\tau=\begin{pmatrix}
		0 & 1\\
		1 & 1
	\end{pmatrix}.
\end{equation}

\section{The search for an energy momentum tensor}\label{sec:em}

We have made some progress on extracting information from a unitary representation of Thompson's groups $F$ and $T$ resembling  conformal data. Indeed, we are able, in special cases, to extract a fusion ring from a representation. However, the goal of producing the conformal data is not complete because we have not yet identified a corresponding \emph{central charge}. This goal is much more challenging, and we'll content ourselves here with outlining the steps required to carry it out.

According to physical arguments a quantum field with scaling dimension
\begin{equation}
	h = 2 = -\log_2(\lambda)
\end{equation}
corresponds to the energy-momentum tensor $T(x)$. However, it is not so simple as that! We also need that under a conformal transformation $w = f(z)$ that $T(z)$ transforms like
\begin{equation}\label{eq:energymomentumtx}
	T(w) = \left(\frac{df}{dz}\right)^{-2}\left(T(z)-\frac{c}{12}\Sch(f,z)\right),
\end{equation}
where $\Sch(f,z)$ is the \emph{Schwarzian derivative}
\begin{equation}
	\Sch(f,z) = \frac{f'''}{f'} - \frac{3}{2}\left(\frac{f''}{f'}\right)^2.
\end{equation}

If we are to suppose that Thompson group elements morally correspond to conformal transformations then if we compare (\ref{eq:energymomentumtx}) with (\ref{eq:nptaction}) we observe a tension between the two transformation laws. Indeed, this tension is what motivates the terminology of ``primary'' for the fields $\hat{\phi}^\alpha(x)$ we introduced. There are two aspects to this: firstly, finding an analogue of $\Sch(f,z)$ for nondifferentiable functions and, secondly, finding operators that exhibit a second Schwarzian-like term at all.  There is one fairly natural candidate for the first problem, namely, exploiting the fact that elements of Thompson's group $T$ may be identified with piecewise linear elements of $\textsl{PSL}(2,\mathbb{Z})$ allows us to realise $f$ as a piecewise $\textsl{PSL}(2,\mathbb{Z})$ map acting on $\mathbb{R}\cup \{\infty\}$. This identification is due to Thurston and is described in detail in \cite{Navas_2011} by Navas. This remarkable identification allows us to realise each element of $T$ as a piecewise function on $\mathbb{R}$ which has Lipschitz continuous first derivative. This, in turn, allows us to calculate the Schwarzian for $f$: we find that $\Sch(f,z)$ is a sum of delta functions (it vanishes on the piecewise parts) at the breakpoints with coefficients determined by the difference of the logarithm of the derivatives of the Thompson group element as a piecewise linear function on $S^1$:
\begin{equation}
  \Sch(f,z) = 2\sum_{z_j\in B_f} (\log_2(f'|_{x_j^+})-\log_2(f'|_{x_j^-}))\delta(z-z_j),
\end{equation}
where $x_j$ is the dyadic rational in $[0,1)$ corresponding to $z_j\in\mathbb{R}$. Thus the Schwarzian of an element of Thompson's group $T$ may be interpreted as a Dirac measure and, upon substituting into a smeared field expression, gives us the desired transformation properties.

The second problem is harder to solve: to find the energy momentum tensor we want to find the generator of infinitesimal conformal transformations. Thus, naively, we want to Taylor expand around a \emph{small} Thomson group element $f\in F$ or $f\in T$ where $f\sim \text{id}+\epsilon$ and write
\begin{equation}
	\pi(f) \approx \mathbb{I} + \pi(\epsilon),
\end{equation}
and identify $T(z)$ with some function of $\pi(\epsilon(z))$. This of course presupposes that our unitary representation $\pi$ is continuous with respect to the standard $L_1$ topology on $F$ (it isn't in general) \cites{jones_no-go_2016,klieschContinuumLimitsHomogeneous2018}. It is not really clear what to do here.

An alternative approach is as follows. We know from CFT that the energy momentum tensor is a \emph{descendent} of a primary field, namely $\mathbb{I}$. What does that mean here? One answer, which we intend to pursue, is to look at ascending operators which extend over several sites. The simplest example of such a thing would be the discrete difference of two ascending operators, e.g.,
\begin{equation}
	\frac{1}{2^m}(\mu^{\alpha}_{j+1} - \mu^\alpha_{j}).
\end{equation}
Depending on $j$ these operators can either vanish, or transform nontrivially, i.e., the action of the ascending channel $\mathcal{E}$ is no longer translation invariant. Suppose we can build something analogous to $\mu^\alpha$ on $k$ contiguous sites. Unfortunately we can no longer expect a nice formula such as (\ref{eq:discretefieldoperator}) for the corresponding field in the continuum limit, so we need to take a somewhat more indirect route to defining such a field operator.

To do so we go back to the original definition (\ref{eq:smearedfieldoperator}) of the smeared field operator and work with smearing functions such as $e^{ikx}$. This is a much more tedious process. One crucial aspect of this approach is that transformation laws such as (\ref{eq:energymomentumtx}) can be possible: additional terms such as the Schwarzian on the right-hand side \emph{can} arise for descendent fields as described above because the extended ascending operators are now sensitive to the presence of breakpoints in elements of Thompson's groups.

\section{Discussion and conclusions}
In this paper we have commenced the study of \emph{Thompson field theory}, the theory of local field-like observables for Jones' semicontinuous limit unitary representations of $F$ and $T$. We have explained how to introduce such fields by renormalising local operators in the observable algebra in such a way that correlation functions converge in the limit of infinite refinement. We also explained how to calculate $n$-point correlation functions for these fields. The short-distance behaviour of the $n$-point functions was explored leading to the identification of an operator product expansion for quasi-primary fields. The transformation laws for $n$-point functions under Thompson group elements was also derived and, in the special case of $\textsl{PSL}(2,\mathbb{Z})$ invariance we reveal a striking analogy to the corresponding laws in conformal field theory.

\section*{Acknowledgements}
Numerous helpful discussions and correspondence with C\'edric B\'eny, Dietmar Bisch, Marcus Cramer, Andrew Doherty, Terry Farrelly, Steve Flammia, Terry Gannon, Jutho Haegeman, Vaughan Jones, Robert K\"onig, Gerard Milburn, Scott Morrison, Emily Peters, Terry Rudolph, Noah Snyder, Tom Stace, Frank Verstraete, Michael Walter, Reinhard Werner, and Ramona Wolf are gratefully acknowledged.

This work was supported by the ERC grants QFTCMPS and SIQS, the DFG through SFB 1227 (DQ-mat), the RTG 1991, the cluster of excellence EXC 201 Quantum Engineering and Space-Time Research, and the Australian Research Council Centre of Excellence for Engineered Quantum Systems (EQUS, CE170100009).

\bibliography{qftg}

\appendix

\section{Some observations concerning trees}\label{app:trees}

Here we collect together some basic observations concerning trees and the circle. Our systems are thought of as living on the circle $S^1$ which is taken to be the interval $[0,1]$ with $0$ and $1$ identified. It is rather convenient to express points $x\in S^1$ in terms of their binary expansions, i.e., we write
\begin{equation}
	x = 0.x_{-1}x_{-2}\cdots x_{-l}, \quad x_{-j} \in \{0,1\}, \quad j = 1, 2, \ldots, l,
\end{equation}
to stand for the representation
\begin{equation}
	x = \sum_{j=1}^l \frac{x_{-j}}{2^{j}},
\end{equation}
for some $l\in \mathbb{Z}^+$.

We introduce the somewhat baffling operation $\ominus$ on $x$ and $y$ in $S^1$ according to
\begin{equation}
	y\ominus x = \sum_{j=1}^l \frac{(y_{-j}-x_{-j})\, \text{mod $2$}}{2^{j}},
\end{equation}
where the arithmetic in the term $y_{-j}-x_{-j}$ is carried out in the finite field $\mathbb{F}_2$ and then embedded back in $\mathbb{R}$ in the natural way.
We pad out the expansions of $x$ or $y$ with zeros as necessary. $x\ominus y$ corresponds to bitwise XOR on the binary digits of $x$ and $y$.

We identify partitions of $S^1$ with trees in the standard way:
\begin{align}
	\{[0,1)\} &\leftrightarrow \mathcal{T}_0 \\
	\{[0,\tfrac12), [\tfrac12,1)\} &\leftrightarrow \mathcal{T}_1 \\
	\{[0,\tfrac14), [\tfrac14,\tfrac12), [\tfrac12,\tfrac34), [\tfrac34,1)\} &\leftrightarrow \mathcal{T}_2 \\
	&\ \vdots,
\end{align}
where $\mathcal{T}_l$ is the regular binary tree with $2^l$ leaves. Each interval in the partition is identified with a leaf of $\mathcal{T}_l$. The nonnegative integer $l$ is called the \emph{level}.

We can alternatively specify a standard dyadic interval $[x,y)= [\tfrac{j}{2^l}, \tfrac{j+1}{2^l})$ by simply writing out the left end point in binary to $l$ significant digits:
\begin{equation}
	[\tfrac{j}{2^l}, \tfrac{j+1}{2^l}) \leftrightarrow 0.x_{-1}x_{-2}\cdots x_{-l}.
\end{equation}
Here the number $l$ of significant digits, the \emph{level}, tells us what kind of standard dyadic interval it is: once you know $x$ you can get $y$ by adding $1/2^l$. Here is a simple example:
\begin{equation}
	[\tfrac{13}{32}, \tfrac{14}{32}) \leftrightarrow 0.01101.
\end{equation}
In this way we can label the leaves of $\mathcal{T}_l$ with binary expansions with exactly $l$ significant digits.

We introduce the following \emph{tree metric} on the leaves of the regular binary tree $\mathcal{T}_l$ as follows. Let $x$ and $y$ be the binary labels corresponding to two leaves of $\mathcal{T}_l$ and recursively define
\begin{equation}
	d_T(x,y) = 1+ d_T(x^{(1)},y^{(1)})
\end{equation}
and
\begin{equation}
	d_T(x,x) = 0, \quad \forall x,
\end{equation}
where
\begin{equation}
	x^{(j)} =  0.x_{-1}x_{-2}\cdots x_{-l+j},
\end{equation}
i.e., by dropping the last $j$ digits of the binary expansion for $x$.
For example, if $x = 13/32$ and $y = 15/32$ we have
\begin{equation}
	d_T(0.01101, 0.01111) = 1+ d_T(0.0110, 0.0111) = 2+d_T(0.011, 0.011) = 2.
\end{equation}

\begin{lemma}\label{lem:treemetric}
The tree metric between $x$ and $y$ in $S^1$ labelling the leaves of $\mathcal{T}_l$ may be computed according to
\begin{equation}
	d_T(x,y) = l+1+\lfloor \log_2(y\ominus x) \rfloor.
\end{equation}
\end{lemma}

As can be seen from the previous example and made rigorous in the lemma, $d_T$ counts, from the right of the binary expansions of $x$ and $y$, the leftmost position at which the digits of $x$ and $y$ are different.

\begin{proof}
	Suppose that $d_T(x,y) = j$. Then we know that $x$ and $y$ share the same first $l-j$ digits, i.e.,
	\begin{equation}
		x = 0.x_{-1}x_{-2}\cdots x_{-l}, \quad \text{and}\quad y=0.x_{-1}x_{-2}\cdots x_{-l+j} y_{-l+j-1}\cdots y_{-l}.
	\end{equation}
	Now notice that
	\begin{equation}
		y\ominus x = 0.00\cdots 0 (y_{-l+j-1}\oplus x_{-l+j-1})\cdots (y_{-l}\oplus x_{-l}).
	\end{equation}
	In particular, note that the digit in the $(l-j+1)$ term is $1$. Hence
	\begin{equation}
		y\ominus x = 0.00\cdots 0 1 \cdots (y_{-l}\oplus x_{-l}) = \frac{1}{2^{l-j+1}}(1+\delta),
	\end{equation}
	where $\delta \in [0,\tfrac{1}{2})$. Take logs of both sides to find
	\begin{equation}
		\log_2(y\ominus x) = -(l-j+1) + \log_2(1+\delta).
	\end{equation}
	Adding $l$ to both sides and taking the floor gives the answer.
\end{proof}

For the special case where $x=0$ and $y = x$ we have the formula
\begin{equation}
	d_T(0,x) = l+1+\lfloor \log_2(x) \rfloor.
\end{equation}

We note the following
\begin{lemma}
	Let $x$ and $y$ be two $l$-digit binary numbers in $[0, 1)$ with $y\ge x$. Then
	\begin{equation}
		y\ominus x \ge y-x
	\end{equation}
	and, hence,
	\begin{equation}
		d_T(x,y) \ge l+1+\lfloor \log_2(|y-x|) \rfloor.
	\end{equation}
\end{lemma}
\begin{proof}
	First note that for $a,b\in \{0,1\}$:
	\begin{equation}
		a-b = \bigl((a-b)\, \text{mod $2$}\bigr) - 2\delta_{a,1}\delta_{b,0},
	\end{equation}
	so that
	\begin{equation}
		a-b \le (a-b)\, \text{mod $2$}.
	\end{equation}
	\begin{equation}
		y\ominus x = \sum_{j=1}^l \frac{(y_{-j}-x_{-j})\, \text{mod $2$}}{2^{j}} = y-x + \delta,
	\end{equation}
	where
	\begin{equation}
		\delta = 2\sum_{j=1}^l \frac{\delta_{x_{-j},1}\delta_{y_{-j},0}}{2^{j}}.
	\end{equation}
	Since $\delta$ is nonnegative we have that
	\begin{equation}
		y\ominus x \ge y-x.
	\end{equation}
\end{proof}

\section{Jordan form for CP maps}\label{app:jordan}

In this appendix we collect together some facts about the Jordan normal form for completely positive maps on $M_n(\mathbb{C})$.

Our observable algebra $\mathcal{A}$ is always a subset of the algebra $\mathcal{B}(\mathcal{H})$ of bounded operators on a finite dimensional Hilbert space $\mathcal{H}$, i.e., $\mathcal{A}\subset M_n(\mathbb{C})$, where $M_n(\mathbb{C})$ is the algebra of $n\times n$ complex matrices. The state space of $\mathcal{A}$ is denoted $\mathcal{S}(\mathcal{A})$, and is given by the set of all positive normalised linear functionals $\omega$ on $\mathcal{A}$.  Any state $\omega$ may be represented by a density operator $\rho \in \mathcal{B}(\mathcal{H})$ via $\omega(A) = \tr(\rho A)$ for all $A$. We have that $\tr(\rho) = 1$ and $\rho \ge 0$.

A \emph{completely positive map} (CP map) $\mathcal{E}:\mathcal{A}\rightarrow \mathcal{A}$ is of the form
\begin{equation}
	\mathcal{E}(X) = \sum_{\alpha} A_\alpha^\dag X A_\alpha, \quad \sum_{\alpha} A_\alpha^\dag A_\alpha = \mathbb{I},
\end{equation}
where $A_\alpha \in \mathcal{B}(\mathcal{H})$, and $\{\alpha\}$ may be chosen finite. We say that $\mathcal{E}$ acts in the \emph{Heisenberg picture}. The dual map $\mathcal{E}^\times : \mathcal{S}(\mathcal{A})\rightarrow \mathcal{S}(\mathcal{A})$ acting on states is given by
\begin{equation}
	\mathcal{E}^\times(\rho) = \sum_{\alpha} A_\alpha \rho A_{\alpha}^\dag.
\end{equation}
The dual map is said to be acting in the \emph{Schr\"odinger picture}.

Define the inner product
\begin{equation}
	(A,B) = \frac1n\tr(A^\dag B),
\end{equation}
called the \emph{Hilbert-Schmidt} inner product. Using $(\cdot, \cdot)$ we obtain a concrete matrix representation of a CP map $\mathcal{E}$ via
\begin{equation}
	[\mathbb{E}]_{\alpha\beta} = (\eta^\alpha, \mathcal{E}(\eta^\beta)),
\end{equation}
where $\eta^\alpha$ is a complete operator basis:
\begin{equation}
	(\eta^{\alpha}, \eta^{\beta}) = \delta^{\alpha\beta}, \quad \alpha, \beta = 1,2, \ldots, n^2.
\end{equation}

We have the following
\begin{proposition}
	If $\mathcal{E}$ is a CP map on $\mathcal{A}$ then its spectral radius satisfies
	\begin{equation}
	r(\mathcal{E}) \le \|\mathcal{E}(\mathbb{I})\|_{\infty}.
	\end{equation}
\end{proposition}

As a matrix $\mathbb{E} \in M_{n^2}(\mathbb{C})$, $\mathbb{E}$ admits a Jordan decomposition of the form
\begin{equation}
	\mathbb{E} = X\left(\bigoplus_{k=1}^K J_k(\lambda_k)\right)X^{-1}, \quad J_k(\lambda) = \begin{pmatrix}\lambda & 1  & \\ & \ddots & 1 \\ && \lambda \end{pmatrix}\in M_{d_k}(\mathbb{C}),
\end{equation}
where $J_k(\lambda)$ are the Jordan blocks of size $d_k$ with $\sum_{k} d_k = n^2$ and the number $K$ of Jordan blocks equals the number of distinct eigenvectors.

The \emph{geometric multiplicity} of an eigenvalue $\lambda$ is equal to the number of Jordan blocks of the form $J_k(\lambda)$. The joint dimension $\sum_{k} d_k \mathbf{I}[\lambda_k = \lambda] $ is called the \emph{algebraic multiplicity} of $\lambda$. The operator $\mathbb{E}$ is said to be non-defective if the geometric multiplicity of every eigenvalue $\lambda$ is equal to its algebraic multiplicity. We always assume, in the sequel, that our operators $\mathbb{E}$ are non-defective.

The Jordan decomposition for $\mathbb{E}$ allows us to infer the existence of left and right eigenvectors for $\mathbb{E}$, i.e., we can write
\begin{equation}
	\mathbb{E} = \sum_{k} \lambda_k |\mu_k )( \nu_k|, \quad (\nu_j|\mu_k) = \delta_{jk}.
\end{equation}

As a linear map on $\mathcal{A}$ we therefore have
\begin{equation}
	\mathcal{E}(A) = \frac{1}{n}\sum_{k} \lambda_k \tr(\nu_k^\dag A)\mu_k, \quad \frac1n\tr(\nu_j^\dag \mu_k) = \delta_{jk},
\end{equation}
so that
\begin{equation}
	\mathcal{E}(\mu_k) = \lambda_k \mu_k, \quad \forall k.
\end{equation}

\end{document}